\documentclass[10pt, conference, compsocconf, final]{IEEEtran}
\usepackage{booktabs} 

\usepackage{graphicx}
\usepackage{epsfig}
\usepackage[tight,footnotesize]{subfigure}
\usepackage{float}

\usepackage{balance}
\usepackage{verbatim}
\usepackage{xspace}
\usepackage{url}
\usepackage{anyfontsize}
\usepackage{hyperref}
\usepackage[normalem]{ulem}
\usepackage[numbers]{natbib}

\usepackage{algorithmic}
\usepackage[algoruled,noend]{algorithm2e}
\usepackage{amssymb, amsmath, amsthm, amsfonts}

\usepackage{color,soul,xspace}
\usepackage{comment}
\newcommand{\todo}[1]{\sethlcolor{red}\hl{TODO: #1}}
\newcommand{\eat}[1]{}

\newtheorem{thm}{Theorem}
\newtheorem{cor}[thm]{Corollary}
\newtheorem{lemma}{Lemma}
\newtheorem*{definition*}{Definition}
\newtheorem{problem}{Problem}
\newtheorem*{obs*}{Observation}

\usepackage{mathtools}
\DeclarePairedDelimiter{\ceil}{\lceil}{\rceil}

\usepackage{tikz}
\usepackage{tikzscale}
\usetikzlibrary{calc}
\usetikzlibrary{shapes}
\definecolor{pinegreen}{cmyk}{0.92,0,0.59,0.25}
\definecolor{royalblue}{cmyk}{1,0.50,0,0}
\definecolor{lavander}{cmyk}{0,0.48,0,0}
\definecolor{violet}{cmyk}{0.79,0.88,0,0}
\tikzstyle{cblue}=[circle, draw, thin,fill=cyan!20, scale=0.8]
\tikzstyle{qgre}=[rectangle, draw, thin,fill=green!20, scale=0.8]
\tikzstyle{rpath}=[thin, red, dashed, opacity=0.1]
\tikzstyle{legend_isps}=[rectangle, rounded corners, thin, 
                       fill=gray!20, text=blue, draw]
                        
\tikzstyle{labelt}=[rectangle, rounded corners, thin,
                           top color= white,bottom color=green!25,
                           black]
\tikzstyle{labelb}=[rectangle, rounded corners, thin,
                          top color= white,bottom color=cyan!25,
                          minimum width=2.5cm, minimum height=0.8cm,
                          royalblue]
\tikzstyle{labelv}=[rectangle, rounded corners, thin,
                          top color= white,bottom color=lavander!25,
                          minimum width=2.5cm, minimum height=0.8cm,
                          violet]

\newcommand{\vsa}{\vspace*{-0.1cm}}
\newcommand{\vsb}{\vspace*{-0.2cm}}
\newcommand{\vsc}{\vspace*{-0.4cm}}

\newcommand{\mouralg}{\text{PHASR}}
\newcommand{\ouralg}{$\mouralg$\xspace}

\begin{document}
\title{Local Community Detection in Dynamic Networks}

\author{
\IEEEauthorblockN{\begin{tabular*}{0.81\textwidth}{@{\extracolsep{\fill} }c c c}
Daniel J. DiTursi\IEEEauthorrefmark{1}\IEEEauthorrefmark{2} & Gaurav Ghosh\IEEEauthorrefmark{1} & 
Petko Bogdanov\IEEEauthorrefmark{1}\end{tabular*}}

\IEEEauthorblockA{\begin{tabular*}{0.72\textwidth}{@{\extracolsep{\fill} }c c}
\parbox[t]{0.35\textwidth}{\centering\IEEEauthorrefmark{1}Department of Computer Science\\
State University of New York at Albany\\
Albany, NY  12222 } &
\parbox[t]{0.35\textwidth}{\centering\IEEEauthorrefmark{2}Department of Computer Science\\
Siena College\\
Loudonville, NY  12211} 
\end{tabular*}}
}
\maketitle

\begin{abstract}
Given a time-evolving network, how can we detect communities over periods of high internal and low external interactions? To address this question we generalize traditional local community detection in graphs to the setting of dynamic networks. Adopting existing static-network approaches in an ``aggregated'' graph of all temporal interactions is not appropriate for the problem as dynamic communities may be short-lived and thus lost when mixing interactions over long periods. Hence, dynamic community mining requires the detection of both the community nodes and an optimal time interval in which they are actively interacting.  

We propose a filter-and-verify framework for dynamic community detection. To scale to long intervals of graph evolution, we employ novel spectral bounds for dynamic community conductance and employ them to filter suboptimal periods in near-linear time. We also design a time-and-graph-aware locality sensitive hashing family to effectively spot promising community cores. 
Our method \ouralg discovers communities of consistently higher quality ($2$ to $67$ times better) than those of baselines. At the same time, our bounds allow for pruning between $55\%$ and $95\%$ of the search space, resulting in significant savings in running time compared to exhaustive alternatives for even modest time intervals of graph evolution. 
\end{abstract}
\vsb
\section{Introduction}
\vsa
\emph{Given a large network with entities interacting at different times, how can we detect communities of intense internal and limited external interactions over a period?} Temporal network interaction data abounds, hence, answering the question above can inform decisions in a wide range of settings.
A set of computers that do not typically interact extensively suddenly exhibits a spike in network traffic---this burst could be the activation of a botnet and warrants special attention from network administrators~\cite{Goel2006}. 
Similarly, prior to a major project deadline, members of a team may communicate more and exclusively among each other compared to other times~\cite{Kleinberg2002}. Knowledge of such periods and the involved parties can lead to better communication systems by accordingly ranking temporally and contextually important messages.

In the example of Fig.~\ref{fig:ex} nodes $\{2,3,4\}$ induce a strong community at times $t_2$-$t_3$ due to multiple high-weight (thick line) internal and weaker external interactions. Intuitively, including more nodes or extending this community in time can only make it less exclusive. Detection of such dynamic communities can be viewed as a generalization of local community detection~\cite{Andersen2006,Andersen2009}, where locality is enforced in both (i) the graph domain: local as opposed to complete partitioning; and (ii) the time domain: communities exhibit bursty internal interactions in a contiguous time interval. For example, in a log of cellular tower interactions in the city of Milan, we detect just such a bursty temporal community near a major highway during hours corresponding to the morning commute. (See Fig.~\ref{fig:Mmap}.)

Unlike evolutionary clustering~\cite{Facetnet,Kim09,berger2010dynamic}, whose goal is to partition all vertices at every timestamp, our goal is to identify the most cohesive communities and their interval of activity without clustering all nodes in time. Hence, our problem is more similar to local community detection~\cite{Andersen2006,Fortunato2010,takaffoli2013incremental} than partitioning. In addition, evolutionary clustering methods are often concerned with the long-term group membership evolution: how partitions appear, grow, shrink and disappear~\cite{berger2010dynamic}. Instead, we focus on interaction bursts among a temporally stable group of nodes.

While local dynamic community detection has practical applications, it also presents non-trivial challenges. Even in static graphs, many local community measures involving cuts are NP-hard to optimize, including conductance~\cite{Sima2006}, modularity~\cite{brandes2006maximizing}, ratio cut~\cite{Wei1989}, and normalized cut~\cite{Shi2000}. 
Furthermore, to detect the active period of a dynamic community one needs to consider a quadratic number of possible intervals. Aggregating all interactions and resorting to static community detection approaches~\cite{Andersen2006,Fortunato2010} may occlude dynamic communities due to mixing interactions from different periods. Alternatively, consideration of individual timestamps in isolation may fragment the community in time. 
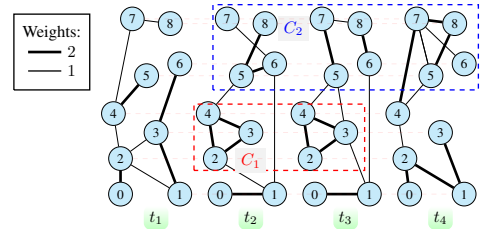
\begin{figure}[t]
\centering
 \resizebox{0.35\textwidth}{!}{%
 \begin{tikzpicture}[auto, thick]

  \foreach \place/\x in {{(-0.55,-2.5)/0}, {(0.7,-2.5)/1}, {(-0.55,-1.75)/2},{(0.2,-1.2)/3},
    {(-0.7,-0.8)/4}, {(0,0)/5}, {(0.7,0.25)/6}, {(-0.3,1.2)/7}, 
    {(0.5,1.1)/8}}
  \node[cblue] (a1\x) at \place {\x};
  \path[ultra thick] (a10) edge (a12);
  \path[thin] (a11) edge (a12);
  \path[ultra thick] (a11) edge (a13);
  \path[thin] (a12) edge (a13);
  \path[thin] (a12) edge (a14);
  \path[thin] (a14) edge (a17);
  \path[ultra thick] (a13) edge (a16);
  \path[ultra thick] (a15) edge (a14);
  \path[thin] (a17) edge (a18);
  
  \foreach \place/\x in {{(1.55,-2.5)/0}, {(2.7,-2.5)/1},{(1.45,-1.75)/2},{(2.2,-1.2)/3},
    {(1.3,-0.8)/4}, {(2,0)/5}, {(2.7,0.25)/6}, {(1.7,1.2)/7},
    {(2.5,1.1)/8}}
    \node[cblue] (a2\x) at \place {\x};
    
  \path[ultra thick] (a20) edge (a21);
  \path[thin] (a21) edge (a26);
  \path[thin] (a21) edge (a22);
  \path[ultra thick] (a22) edge (a23);
  \path[ultra thick] (a23) edge (a24);
  \path[ultra thick] (a22) edge (a24);
  \path[ultra thick] (a25) edge (a26);
  \path[thin] (a25) edge (a24);
  \path[thin] (a26) edge (a27);
  \path[ultra thick] (a25) edge (a28);
    
  \foreach \place/\x in {{(3.55,-2.5)/0}, {(4.7,-2.5)/1},{(3.45,-1.75)/2},{(4.2,-1.2)/3},
    {(3.3,-0.8)/4}, {(4,0)/5}, {(4.7,0.25)/6}, {(3.7,1.2)/7},
    {(4.5,1.1)/8}}
    \node[cblue] (a3\x) at \place {\x};
    
  \path[ultra thick] (a30) edge (a31);
  \path[thin] (a31) edge (a33);
  \path[thin] (a31) edge (a36);
  \path[ultra thick] (a32) edge (a33);
  \path[ultra thick] (a33) edge (a34);
  \path[ultra thick] (a32) edge (a34);
  \path[thin] (a34) edge (a33);
  \path[thin] (a35) edge (a34);
  \path[thin] (a35) edge (a33);
  \path[ultra thick] (a35) edge (a37);
  \path[thin] (a37) edge (a38);
  \path[ultra thick] (a36) edge (a38);
    
  \foreach \place/\x in {{(5.55,-2.5)/0}, {(6.7,-2.5)/1}, {(5.45,-1.75)/2},{(6.2,-1.2)/3},
    {(5.3,-0.8)/4}, {(6,0)/5}, {(6.7,0.25)/6}, {(5.7,1.2)/7}, 
    {(6.5,1.1)/8}}
  \node[cblue] (a4\x) at \place {\x};
  
   \path[ultra thick] (a40) edge (a42);
 \path[ultra thick] (a41) edge (a42);
  \path[thin] (a47) edge (a45);
  \path[ultra thick] (a41) edge (a43);
  \path[thin] (a42) edge (a44);
  \path[ultra thick] (a44) edge (a47);
  \path[thin] (a47) edge (a46);
  \path[ultra thick] (a45) edge (a48);
  \path[thin] (a45) edge (a44);
  \path[thin] (a45) edge (a47);
  \path[ultra thick] (a47) edge (a48);
    
 \draw[red,thick,dashed] (1,-2) --  (1,-0.6) -- (4.6,-0.6) -- (4.6,-2) -- (1,-2);
 \node[red, fill=gray!10] at (2.2,-1.8){$C_1$};
 \draw[blue,thick,dashed] (1.4,-0.3) --  (1.4,1.5) -- (7,1.5) -- (7,-0.3) -- (1.4,-0.3);
 \node[blue, fill=gray!10] at (3.1,1){\bf{$C_2$}};
 
 \node[] at (-2,0.9){Weights:};
 \draw[black,ultra thick] (-2.5,0.5) -- (-1.8,0.5); \node[right] at (-1.8,0.5){$2$};
 \draw[black,thick] (-2.5,0.1) -- (-1.8,0.1); \node[right] at (-1.8,0.1){$1$};
 \draw (-2.8,-0.3) --  (-1.2,-.3) -- (-1.2,1.3) -- (-2.8,1.3) -- (-2.8,-.3);
 
 
  \edef\jpo{0}
  \foreach \i in {1,...,8}
  {
    \foreach \j in {1,...,3}
    {
      \pgfmathparse{int(\j+1)}
      \edef\jpo{\pgfmathresult}
      \path[rpath] (a\j\i) edge (a\jpo\i);
    }
   }
 
  \node[labelt] at (0.2,-3){$t_1$};
  \node[labelt] at (2.2,-3){$t_2$};
  \node[labelt] at (4.2,-3){$t_3$};
  \node[labelt] at (6.2,-3){$t_4$};
\end{tikzpicture}
}\vsb
\caption{Interactions in a small graph over $4$ time steps. Thicker edges designate stronger interactions. Two temporal communities $C_1$ across $t_2$-$t_3$ and $C_2$ across $t_2$-$t_4$ are designated in dashed boxes.} 
\label{fig:ex}
\vsb\vsb\vsa
\end{figure}
\begin{figure}[t]
\centering
\hspace{-0.1in}
\includegraphics[width=0.45\textwidth]{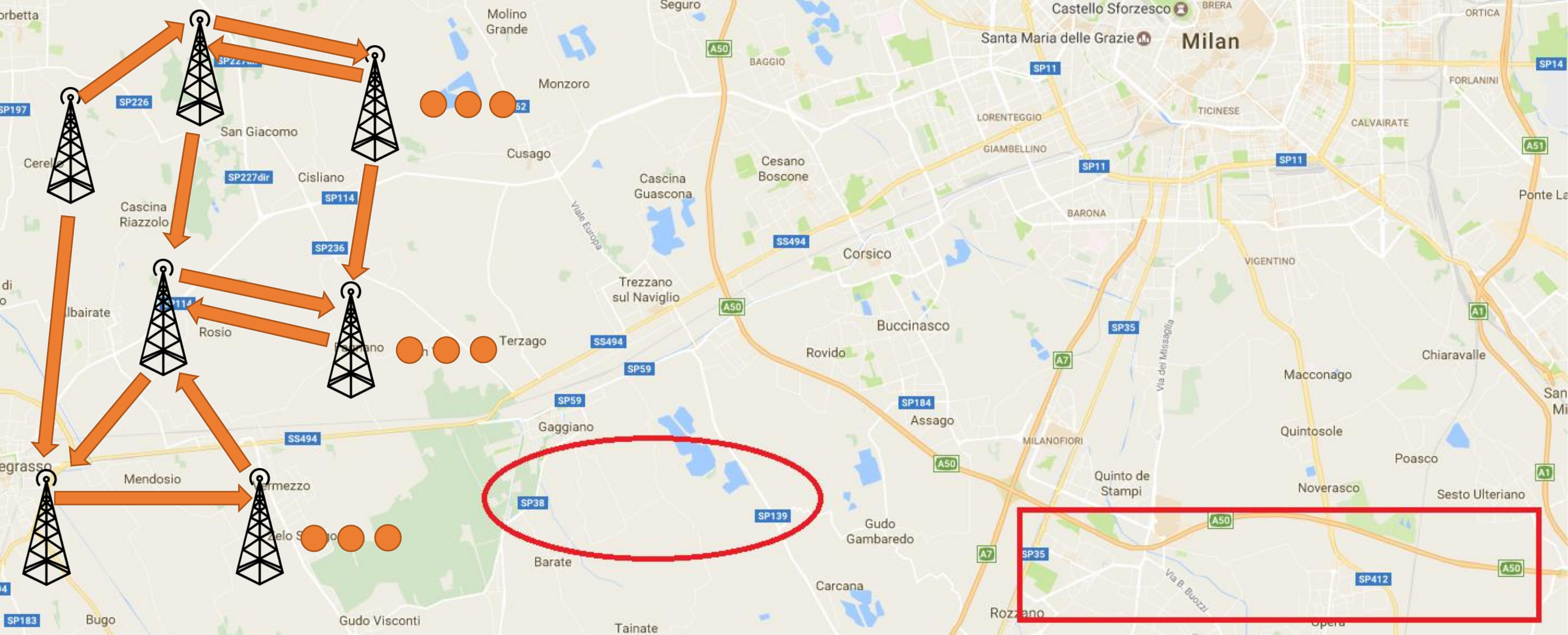}\vsa
\caption{Map of Milan and surroundings. The marked areas contain the cell towers of the top discovered communities in the {\em Call\em} data set.  These communities exist in the early morning hours and likely correspond to morning commute along Milan's beltway.}\vsb\vsb
\label{fig:Mmap}
\end{figure}

We propose a \emph{Prune, HASh and Refine (\ouralg)} approach for the problem of \emph{temporal community detection}. We \emph{prune} infeasible time periods based on novel spectral lower bounds for the graph conductance tailored to the dynamic graph setting. We show that pruning all $O(T^2)$ possible intervals can be performed in time $O(T\log{T})$ due to an interval grouping scheme exploiting the similarity of overlapping time intervals. In order to efficiently \emph{spot} candidate community nodes in non-pruned time intervals, we design a \emph{time-and-graph-aware} locality sensitive hashing scheme to group similar temporal neighborhoods of community nodes in linear time. Our hashing scheme can be configured to maximize the probability of spotting communities of a target duration informed by the pruning step. In the \emph{refinement} step, we expand communities rooted in candidate nodes in time.
Our contributions in this work are as follows:\\
    \noindent $\bullet$ We propose novel spectral bounds for temporal community conductance and an efficient scheme to compute them using $O(T\log{T})$ eigenvalue computations. Our bounds enable pruning of more than $95\%$ of the time intervals in synthetic and more than $50\%$ in real-world instances; and are trivial to parallelize.  \\
    \noindent $\bullet$ We propose a joint time-and-graph, locality-sensitive family of functions and employ them in an effective scheme for spotting temporal community seeds in linear time. Our LSH scheme enables the discovery of 2 to 67 times better communities compared to those discovered by baselines. \\
    \noindent $\bullet$ \ouralg scales to synthetic and real-world networks of sizes that render exhaustive alternatives infeasible. This dominating performance is enabled by effective pruning and high true positive rate of candidates produced by our hashing scheme.

\vsb
\section{Related work}
\vsa
\noindent{\bf Static local communities:} Our work is different from static community detection~\cite{xie2013overlapping,Andersen2006,macropol2010scalable,Fortunato2010,leskovec2010empirical} in that we consider a dynamic setting. We compare to temporal generalization of the method in~\cite{Andersen2006} since it similarly focuses on the subgraph conductance of a community and the method in~\cite{macropol2010scalable} as it employs hashing for static network communities. Our experiments demonstrate that naive extensions of the above methods do not scale well with time. 

\noindent{\bf Dense subgraph} detection has been recently considered for unweighted interactions in time~\cite{rozenshtein2014,gaumont2016finding, rossetti2016tiles}. Such methods allow for analysis at the maximal temporal resolution in which only fragments of the community may be available at an instance and thus require parameters that predefine the total span of a community~\cite{rozenshtein2014} or some notion of persistence (e.g. time-to-live interval) for interaction edges~\cite{gaumont2016finding, rossetti2016tiles}. In addition, these works consider the density within a community, but not its separation from the rest of the network, i.e. its cut. In contrast, in our setting we focus on communities that are well-separated from the rest of the network, we do not require predefined persistence/span, and we allow for weights of edges modelling varying strength of interactions in time. 

\noindent{\bf Persistent subgraphs} in time have also been considered~\cite{liu2014persistent,Ahmed2011}. Similar to this work, the methods in this category require that the subgraph of interest persists as either a conserved topology~\cite{Ahmed2011} or as a stable level of intra-community interactions. These objectives are different and even complementary to ours in that they do not consider how well-separated the community is from the rest of the network. The persistence requirements can be considered in conjunction with conductance to further rank candidates of interest.   

\noindent{\bf High-weight temporal subgraph} detection is another relevant setting with the goal of detecting connected subgraphs which optimize a function of their node or edge weights in time have also been considered~\cite{liu2014persistent,Ahmed2011,MEDEN}. Similar to this work, the methods in this category require that the subgraph of interest persists as either a conserved topology~\cite{Ahmed2011} or as a stable level of intra-community interactions. These objectives are different from ours in that they do not consider how well-separated the community is from the rest of the network; the focus is only on high internal weights. The persistence requirements could be considered in conjunction with conductance to further rank candidates of interest.   

\noindent{\bf Evolutionary clustering }for dynamic networks is another related and very active area of research~\cite{Facetnet,Kim09,berger2010dynamic,takaffoli2013incremental}. The goal in evolutionary clustering is to track the changes in the global network partitions over time, where partitions are allowed to vary from one time slice to the next by incorporating temporal smoothness of partition membership. The general problem setting, however, differs from ours as the goal is to partition all vertices at every timestamp as opposed to identifying the best local communities and their interval of activity. Closest to our goal from this group is the local community method by Takafolli et al.~\cite{takaffoli2013incremental} which extends local communities of good modularity from one time-step to the next. In comparison, our approach considers the full timeline as opposed to only consecutive time steps, and as a result consistently finds lower-conductance communities than that of Takafolli et al.~\cite{takaffoli2013incremental}.

\vsb
\section{Problem definition}
\vsa
Our goal is to find communities of stable membership over a period of time during which members interact mostly among each other as opposed to with the rest of the network. To model this intuition we propose the \emph{temporal conductance} measure, a natural extension to graph conductance which is commonly adopted for local communities in static graphs~\cite{Sima2006,Andersen2006}. Our problem can then be cast as detecting the subgraph and interval of smallest temporal conductance, formalized next. 

Let $G(V,E,W)$ be an undirected edge-weighted {\em temporal graph\em}, where $V$ is the set of vertices, $E \subseteq V \times V$ is the set of edges and $W$ is a family of weight functions $W: E \times T \rightarrow \mathbb{R}^+$ mapping edges to real values across a discrete {\em timeline\em} $\mathcal{T}=\{0,1,\ldots,|T|-1\}$ of graph evolution. We will use $w(u,v,t)$ to denote the weight of an edge $(u,v)$ at time $t$ and $w(u,v,t,t')=\sum_{i=t}^{t'}w(u,v,i)$ to denote the aggregate (temporal) weight on the same edge in the interval $[t,t']$. The temporal volume of a node $u$ is defined as $vol(u,t,t')=\sum_{(u,v)\in E}w(u,v,t,t')$. 
A {\em temporal community\em} $(C,t,t')$ is a connected subgraph of $G$ induced by nodes $C \subseteq V$ and weighted by $w(u,v,t,t'), \forall (u,v) \in E \cap (C \times C)$. 
The {\em temporal conductance,\em} of a community $(C,t,t')$, a generalization of the classic conductance~\cite{Sima2006}, is defined as: \vsa
$$ \phi(C,t,t') = \eta(t,t')\frac{cut(C,t,t')}{\min(vol(C,t,t'), vol(\bar{C},t,t'))},\vsb$$
where $\bar{C}=V\setminus C$ ; 
$cut(C,t,t') = \sum_{u \in C,v \in \bar{C} }w(u,v,t,t')$
is the temporal cut of $C$; and $\eta(t,t')$ is a temporal normalization factor. The smaller the conductance, the more cohesive the community. 

If $\eta(t,t')$ is a constant, the temporal conductance reduces to the regular graph conductance of $C$ on an aggregated network in the interval $[t,t']$. However, without normalization the conductance will favor small communities in single timestamps, thus fragmenting a natural community in time. Hence, we consider a temporal normalization $\eta(t,t')=(t'-t)^{-\alpha}$, where $\alpha$ controls the importance of community time extent. Our methods can trivially accommodate different forms of the normalization function, e.g. exponential time decay similar to that used in streaming settings~\cite{sharan2008temporal,yu2013anomalous}. 

To demonstrate the effect of normalization, consider $C_1$ and $C_2$ in Fig.~\ref{fig:ex}. When $\alpha=0$ (i.e. no normalization), the conductance of $C_1$ is $\phi(C_1,2,3)=5/29=0.17$, while $\phi(C_2,2,4)=8/39=0.21$ (weights considered). Upon increasing the normalization (say $\alpha=1$), and hence the preference for longer-lasting communities, the temporal conductance of $C_2$ becomes lower than that of $C_1$.



\begin{problem}\vsa
{\bf [Lowest temporal conductance community]} \\ Given a dynamic network $G(V,E,W)$, find the community: \\$(C_o,t_o,t_o')=\arg\min_{C\in V, 0\leq t\leq t'\leq T}\phi(C,t,t')$.\vsa
\end{problem}
The \emph{lowest conductance} problem in a static graph is known to be NP-hard~\cite{Sima2006} and since a dynamic graph of a single timestamp $T = 1$ is equivalent to the static case, our problem of \emph{temporal conductance minimization} is also NP-hard. Hence, our focus is on (i)~scalable processing of dynamic graphs over long timelines; and (ii)~effective and efficient detection of community seeds in the graph and time which existing approximate solutions for the static case require as input~\cite{Andersen2006,Andersen2009}.



\vsb
\section{Methods}
\vsa
Since our problem is NP-hard, exhaustive approaches would not scale to large real-world dynamic networks. We experimentally demonstrate that na{\"i}ve heuristics are infeasible in all but trivially-small instances. Thus, our overall method enables scalability (i)~with the graph size by identifying and refining candidate seed nodes in time that are likely to participate in low-conductance communities (\emph{Seed selection~\ref{sec:seed}}); and (ii)~with the length of the timeline by pruning infeasible periods in time with guarantees (\emph{Pruning Sec.~\ref{sec:prune}}). Our final approach is presented in Sec.~\ref{sec:alg}. 



\vsb
\subsection{Preliminaries}
\vsa
\label{sec:prelim}
Before we present our solution, we review preliminaries related to \emph{locality sensitive hashing (LSH)}~\cite{indyk1998approximate} and \emph{local community detection}~\cite{Andersen2006,Avron2015}. Both concepts have been employed for static networks and are, thus, a natural starting point for na{\"i}ve baselines.

\noindent{\bf LSH and neighborhood similarity.} 
Indyk et al.~\cite{indyk1998approximate} proposed LSH for approximate nearest neighbor (NN) search. A family of functions is $(d_1,d_2,p_1,p_2)${\em -sensitive\em} w.r.t. a distance measure $d$ if for every function $f$ in the family, $d_2\leq d(x,y) \leq d_1\Rightarrow p_2\geq P[f(x) = f(y)] \geq p_1$. 
The family of \emph{minhash} functions for sets was shown to be locality-sensitive w.r.t. the Jaccard distance (defined as $1 - $Jaccard Similarity)~\cite{broder1998min,indyk1998approximate}. 
The concept of LSH for node neighborhoods was employed for fast community detection in static networks by Macropol et al.~\cite{macropol2010scalable}. The intuition is that nodes in dense communities tend to have similar neighborhoods and thus they will collide when hashed using an LS family. We extend this intuition to locality in time by considering similarity in both time and graph space. In addition, the community strength in time is dependent not only on the existence of edges but also the level of interaction, which we model as weights in time $w(u,v,t)$. In order to incorporate weights we adopt a recent approach by Ioffe et al.~\cite{ioffe2010improved} for LSH of weighted sets using a weighted Jaccard similarity, defined as follows for neighborhoods in our setting:\vsb
$$ J_W(N_i,N_j) = \frac{\sum_{v \in N_i \cup N_j} min(w(v,i),w(v,j))}{\sum_{v \in N_i \cup N_j} max(w(v,i),w(v,j)) }, \vsb$$
where $N_i$ and $N_j$ are the neighborhoods of nodes $i$ and $j$ in a given time period (time indices of the weight omitted for simplicity). 

\noindent{\bf Local communities in static graphs.} Recent approaches for low-conductance local community detection in static networks rely on graph diffusion~\cite{Andersen2006,Avron2015}. The goal is to obtain a single partition around a predefined seed node by a local computation that involves a small fraction of the graph around the seed and that closely approximates the best conductance involving the seed. In our proposed approach, we first ``spot'' seed nodes in time based on high temporal neighborhood similarity and expand to a community similar to the spectral sweep method by Andersen et al.~\cite{Andersen2006}.



\vsb
\subsection{Temporal Neighborhood LSH to spot seeds}
\label{sec:seed}
A low-conductance temporal community consists of nodes that mostly interact with each other over a contiguous time interval when the community is active. A first step in our approach is to find seed nodes within the community in the corresponding time frame that can then be used to expand to strong communities. Our solution for seed selection is based on the observation that weighted node neighborhoods within the community tend to be similar. We exploit this observation in order to obtain seeds for promising regions in time. Specifically, we propose a scalable similarity search method based on hashing of node neighborhoods in time. We show that our scheme is locality-sensitive in both time and graph space. Its parameters can be optimized to target a pre-specified duration in time---a property that we exploit in conjunction with temporal pruning of feasible intervals in order to reduce the computational and memory footprint of our approach.

Given a dynamic graph $G(V,E,W)$, the weighted temporal neighborhood $N_u^t$ of node $u$ is the weighted set of its neighbors (including $u$): $N_u^t=\{(v:w(u,v,t)) \vert (u,v)\in E\}\cup (u:vol(u,t))$. We adopt the weighted Jaccard similarity $J_W(N_u^t,N_v^t)$ and the weighted minhash function for $\psi(\cdot)$ ensuring that 
\vsb$$P[\psi(N_u^t)=\psi(N_v^t)]=J_W(N_u^t,N_v^t).\vsb$$ We hash neighborhoods using $r$ independent minhash functions to create a graph signature $S_G^r(N_u^t)$ for a given neighborhood $N_u^t$. 

If we compare weighted neighborhoods, disregarding the time at which they were observed, we may produce collisions of high-similarity neighborhoods that may be potentially distant in time. Hence, a straightforward adoption of LSH for weighted sets will not be locality sensitive with respect to time. Instead, we need to associate highly similar neighborhoods that are also close in time. 
The main intuition behind our temporal locality sensitive hashing function is: \emph{close time instants are likely to belong to the same interval if the timeline is partitioned into random segments}. 

Let $p^k=\{p_1 < p_2 \dots < p_k\}$ be a $k$-partitioning of the timeline using $k$ pivot time points selected uniformly at random in $[0,T]$. We define a hash function $\tau^k(\cdot)$ based on the partitioning $p$ that maps a given time point $t$ to the index of the earliest pivot $p_i \in p$ whose time exceeds $t$:
$\tau^k(t)=\{min(i)|p_i\geq t,p_i\in p^k\}.$ 
\vsb
\begin{thm}{\bf[Temporal locality]}\label{thm:tlocal}
The family of temporal hash functions $\tau^k$ is $(\Delta_1,\Delta_2,(1-\frac{\Delta_1}{T})^k,(1-\frac{\Delta_2}{T})^k)${\em -sensitive\em} family for the distance in time $\Delta$ defined as the delay between two timepoints.
\end{thm}

\begin{proof}\vsa
Available in the Appendix. \vsb
\end{proof}

Beyond being locality sensitive in time, our pivot-based hashing family $\tau^k$ can be configured to target specific community lengths. While we do not know the duration of good communities in the data a priori, as we will show in the following section, we can prune intervals in time that cannot include the best communities with guarantees. Hence, we need to be able to focus on matching neighborhoods at time resolutions that are viable in order to reduce the memory and running time footprint of our solution. To enable this, we need to answer the following question: \emph{What is the optimal number of pivots $k$ to detect communities of a given duration?} Assuming that a target community duration is $\Delta^*$, we need to choose $k$ such that similar temporal neighborhoods within that period have a high chance of collision. A perfect partitioning $p$ of the timeline would produce a single segment that isolates the target period of length $\Delta^*$. This requires two pivots to ``bracket'' the period and all other pivots to fall outside of it.
\vsb
\begin{thm}{\bf[Optimal number of pivots]}
The number of pivots $k^*$ that maximizes the probability of a perfect partition of a period of length $\Delta^*$ is $k^*\approx\lfloor\frac{2T}{\Delta^*}\rfloor$.
\label{thm:optk}
\end{thm}
\begin{proof}\vsa
Available in the Appendix. \vsb
\end{proof}


To detect similar neighborhoods in time we combine an $r$-sized graph signature $S^r_G(N_u^t)$ with a temporal signature $S^k_T(N_u^t)=\tau^k(t)$ using an \emph{AND} predicate to obtain a unified {\em temporal neighborhood signature \em} $S^{r,k}(N_u^t)$ that is a locality sensitive family in both the time and graph domains as a direct consequence of the locality $S^r_G(N_u^t)$ and $S^k_T(N_u^t)$.
\vsa
\begin{cor}
Let $S^{r,k}$ be a temporal neighborhood hash function with $r$ minhashes and $k$ partitions. Then: $P[S^{r,k}(N_u^t)=S^{r,k}(N_u^{t'})]=J_W(N_{u}^{t},N_{v}^{t'})^r (1 - \frac{|t - t'|}{T})^k.$\vsa
\end{cor}

Since our composite hashing family is locality sensitive in both time and the graph, we can construct signatures that amplify its locality sensitivity. For example, if we combine $l$ independent hash signatures $S^{r,k}$ using an \emph{OR} predicate, i.e. require that there is a match in at least one hash value for a collision, the resulting collision probability will be $1-[1-p_{r,k}]^l$, where $p_{r,k}=P[S^{r,k}(N_u^t)=S^{r,k}(N_u^{t'})]$. Similarly, an \emph{AND} predicate composition will result in $p_{r,k}^l$ probability of collision. Using cascades of such composition we can ``shape'' the selectivity of our LSH scheme and thus control rate of FP and FN collisions at the cost of increased memory and computational overhead. 

A neighborhood signature requires $\log (k|V|^r)$ bits of storage, since a single weighted minhash value is a vertex index and a temporal hash value is the position of the first pivot index exceeding the timestamp of the hashed neighborhood. For a fixed temporal resolution and a composition of $b$ independent signatures (OR/AND predicate compositions), the overall memory footprint of our hashing approach will be $bT|V|\log(k|V|^r)$ bits. The above analysis is pessimistic as it assumes no collisions and thus storing the signatures for all existing collision bins. Nevertheless, large and long-evolving instances may be impractical to exhaustively hash and process. Our filtering approach discussed next addresses this challenge and allows us to significantly reduce the memory and computational footprint.

\vsb
\subsection{Spectral bounds for pruning time intervals}
\label{sec:prune}
The strongest temporal communities exhibit lower conductance than other communities in the network. Also, in real-world graphs many time intervals contain no promising communities---i.e. no project deadline or spike in network traffic. Our goal is to eliminate from consideration such periods of low community activity which cannot coincide with the best temporal community. In what follows, we develop lower bounds on the temporal conductance of any subgraph in a given time window. We employ our bounds in combination with a solution estimate to deterministically prune irrelevant intervals in time. Such pruning can significantly improve the running time of our LSH-based approach as we can target only promising neighborhoods in time by adjusting the time scale (i.e. number of pivots $k$) for temporal hashing. 

Let $G^{[t,t']}$ be the aggregate graph of $G(V,E,W)$ over time interval $[t,t']$ with aggregate edge weights $w(u,v,t,t')$. The temporal \emph{graph conductance} in an aggregate weighted graph is defined as the minimum temporal conductance over all subsets $C\in V$:
$
\phi(G^{[t,t']}) = \min_{C\in V}\phi(S,t,t').
$
Let $A$ be the adjacency matrix of a weighted graph $G^{[t,t']}$ with elements $A_{u,v}=w(u,v,t,t')$ and $D$ be the diagonal ``degree'' matrix with elements $D_{u,u}=vol(u,t,t')$ and $0$ in all off-diagonal elements. The matrix $\mathcal{L}=D-A$ is the unnormalized graph Laplacian, while the matrix $\mathcal{N}=D^{-1/2}\mathcal{L}D^{-1/2}$ is the symmetric normalized graph Laplacian ~\cite{Spielman2010}. 
The Laplacian matrices have many advantageous properties and have been employed in spectral graph partitioning~\cite{Spielman2010,chung1997spectral}. The eigenvalues $0=\lambda_1\leq\lambda_2\leq...\leq\lambda_{|V|}\leq2$ of $\mathcal{N}$ are all real, non-negative and contained in $[0,2]$. The smallest eigenvalue is $0$ and its multiplicity is the same as the number of connected components. Assuming that the graph is connected (i.e. one connected component), one can show the following relationship with the graph conductance:
\vsb
\begin{lemma}{\bf [Spectral bound~\cite{Spielman2010}]} The temporal graph conductance of a weighted graph can be bounded as follows:
$
\uline{\phi}(G^{[t,t']}) = \eta(t,t')\lambda_2/2 \leq \phi(G^{[t,t']}).
$
\end{lemma}\vsa
Note that the above bound is valid for arbitrary weighted graphs, although we explicitly state it in the context of aggregated graphs including the normalization based on $\eta(t,t')$. The conductance of any approximate solution $\bar{\phi}$ can serve as an upper bound to that of the lowest conductance in $G(V,E,W)$ and can be employed to prune irrelevant intervals: 
\vsb
\begin{cor} {\bf[Pruning]}
If $\bar{\phi} \leq \eta(t,t')\lambda_2(\mathcal{N}^{t,t'})/2$, then $[t,t']$ does not contain the lowest conductance temporal community.\vsb
\end{cor}

The corollary follows directly from the spectral bound. An intuitive approach for pruning is to compute $\eta(t,t')\lambda_2(\mathcal{N}^{t,t'})/2$ for the aggregated graphs of all possible intervals in time, incurring a quadratic number of eigenvalue computations which will not scale to large graphs evolving over long periods of time. In what follows, we show that one can obtain a lower bound for $\lambda_2$ of an interval based on the eigenvalues in sub-intervals reducing the number of necessary eigenvalue computations in our pruning strategy. 

\vsb
\begin{lemma}
\label{lem:dom}
Let $A$ be a real positive semi-definite matrix of dimension $n$, and let $d$ and $\epsilon$ be two real vectors of the same dimension s.t. $d_i\geq0, \epsilon_i\geq 0\ \forall i\leq n$. Then, 
$\min_{f\perp d+\epsilon}f^TAf \geq \min_{g\perp d} g^TAg.$ \vsb
\end{lemma}
\begin{proof}
Available in the Appendix.
\end{proof}\vsb

\begin{thm}
\label{thm:sumbound} {\bf[Composite bound]}
Let $[t,t']$ be partitioned in $k$ consecutive non-overlapping subintervals $\{[t_1,t_1'], [t_2,t_2']...[t_k,t_k']\}$ such that $t_i=t_{i+1}'-1, \forall i\in[1,k)$ with corresponding aggregated normalized graph Laplacians $\mathcal{N}_i$. Then,\vsa
$$\lambda_2(\hat{\mathcal{N}}) \geq \sum_{i=1}^{k} \min_{u\in V}\frac{vol(u,t_i,t_i')}{vol(u,t,t')}\lambda_2(\mathcal{N}_i),\vsb$$
\noindent where $\lambda_2(\mathcal{N}_i)$ is the second smallest eigenvalue of $\mathcal{N}_i$, and $\hat{\mathcal{N}}$ is the Laplacian of $G^{[t,t']}$.\vsb 
\end{thm}
\begin{proof}
Available in the Appendix.
\end{proof}\vsb\vsa

The composite bound for $\lambda_2(\hat{\mathcal{N}})$ enables pruning intervals without explicitly computing their interval eigenvalues. 
Given any partitioning $\{[t_1,t_1']...[t_k,t_k']\}$ of $[t,t']$, we can prune using the \emph{composite bound} $\uline{\phi_c}(G^{[t,t']}) = \eta(t,t')\sum_{i=1}^{k} \min_{u\in V}\frac{vol(u,t_i,t_i')}{vol(u,t,t')}\lambda_2(\mathcal{N}_i).$ 

To enable scalable pruning, we can pre-compute $\lambda_2$ for a subset of intervals and attempt to prune all intervals using $\uline{\phi_c}$ instead of an exhaustive eigenvalue computation. There is a trade-off between how many eigenvalues to pre-compute and the pruning power of $\uline{\phi_c}$. If we only compute single-time snapshot intervals, we can obtain all composite bounds, however, they may not be very tight for longer intervals. If we pre-compute too many intervals, we will incur cost similar to the exhaustive all-eigenvalue computation. 

We adopt a multi-scale scheme in which we pre-compute non-overlapping intervals of exponentially increasing lengths: \vsa $$\vsa l^i,l\in\mathbb{N}_{\geq2},\forall i\in\mathbb{N}_{[0,\ceil{log(T)}]}.$$ For example, if $l=2$, we compute $\lambda_2$ for non-overlapping intervals of sizes powers of $2$, i.e. $\{[0,0],\ldots[T,T],[0-1],[2-3],\ldots[T-1,T],\ldots\}$. The pre-computation requirement for our composite bound scheme is $O(Tlog(T)B),$ where $B$ is the time to compute $\lambda_2$ for a single aggregated graph using the Lanczos method. To compute $\uline{\phi_c}$ for any interval we incur cost $O(|E|log(T))$ as any interval can be composed by at most $O(log(T))$ subintervals with known eigenvalues. 

While our composite scheme requires sub-quadratic (in $T$) eigenvalue computations, we still need to compute $\uline{\phi_c}$ for $O(T^2)$ intervals to prune them. To further speed up the process, we group intervals of significant overlap and attempt to prune using a group-level bound without composing individual intervals within the group. To this end, we define a {\em pruning group\em} $\tau=(t,t',t'')$ as a set of intervals  with a common start $t$ and ending at times $t'$ to $t'', t'<t''$. We ensure a significant overlap between all group members by enforcing that the common interval prefix exceeds a fixed fraction of the length of all group members: $\frac{t'-t}{t''-t}\geq\beta$. Given a partitioning $\{[t_i,t'_i]\}$ of the group prefix $[t,t']$, we define the group lower bound as $\uline{\phi_c}(G^\tau)=\eta(t,t'')\sum_{i=1}^{k} \min_{u\in V}\frac{vol(u,t_i,t'_i)}{vol(u,t,t'')}\lambda_2(\mathcal{N}_i)$. The differences from the composite bound of the prefix $\uline{\phi_c}(G^{[t,t']})$ is in (i) the denominator of the fraction 
and (ii) the normalization $\eta(t,t'')$ on the RHS.\vsa 
\begin{thm} \label{thm:groupbound} {\bf[Group composite bound]}
Let $\tau=(t,t',t'')$ be a group of shared-prefix intervals, then $\uline{\phi_c}(G^\tau)\leq\uline{\phi_c}(G^{[t,t^*]}), \forall t^*\in [t',t''].$\vsa
\end{thm}
\begin{proof}
Available in the Appendix.
\end{proof}\vsb\vsa
\eat{
prefix that is at least $50\%$ of the total length of each interval---that is:
$$P_{[a,c]} = \{[a,b],[a,b+1],...,[a,c]\}$$
where $b = \ceil{(a+c)/2}$. For any $I=[a,b+\tau]\in P_{[a,c]}$, we partition the interval $[a,b]$ into $\big\{[t_1,t'_1]\ldots[t_k,t'_k]\big\}$ as before and then treat $[b+1,b+\tau]$ as a single subinterval, giving:


$$\uline{\phi_c}(G^I) = \eta(a,b+\tau)\Big(\sum_{i=1}^{k} \min_{u\in V}\frac{vol(u,t_i,t'_i)}{vol(u,a,b+\tau)}\lambda_2(\mathcal{N}_i) + $$
$$\min_{u\in V}\frac{vol(u,b+1,b+\tau)}{vol(u,a,b+\tau)}\lambda_2(\mathcal{N}_{[b+1,b+\tau]})\Big)$$
\begin{lemma}
For all $I\in P_{[a,c]}$,
$$\eta(a,c)\sum_{i=1}^{k} \min_{u\in V}\frac{vol(u,t_i,t'_i)}{vol(u,a,c)}\lambda_2(\mathcal{N}_i) \leq \uline{\phi_c}(G^I)$$
\end{lemma}
This follows because $[a,c]$ is at least as long as $[a,b+\tau]$; therefore $vol(u,a,c) \geq vol(u,a,b+\tau)$ and ${\eta(a,c) \leq \eta(a,b+\tau)}$. $\qed$

If we have an estimate smaller than the quantity above, we can discard the entire pruning group at once.
}
\vsb
\subsection{\ouralg: Prune, HASh and Refine}
\vsa
\label{sec:alg}



\begin{algorithm}[t]
\begin{algorithmic}[1]
\caption{\ouralg}
\footnotesize
 \REQUIRE $G(V,E,W)$, $\alpha$, LSH rows $r$, bands $b$, pruning res. $l$
 \ENSURE A set of temporal communities $\mathcal{C}=\{(C_i,t_i,t'_i)\}$
 \STATE Compute bounds $\Phi$ at scales $l^i,i=0\ldots\ceil{log(T)}$ 
 \STATE Compute an estimate $\phi^*$ using $\Phi$
 \STATE Prune intervals $[t,t']\in\tau$ based on $\uline{\phi_c}(G^\tau)\geq\phi^*$
 \STATE Prune remaining intervals $[t,t']$ based on $\uline{\phi_c}(G^{[t,t']})\geq\phi^*$ 
 \FORALL{$(u,t) \in (V,[1\ldots T])$}
   \FORALL{scales $s\in{1\ldots T/2}$}
     \IF { $\exists$ an unpruned $[l,r]\in[t-s,t+s],$ } 
       \STATE \emph{Hash($N_u^t,k^*(s),r,b$)}
     \ENDIF
    \ENDFOR
 \ENDFOR
 \FOR{$\forall$ Buckets $B$ sorted by \emph{fill-factor}}
   \STATE $[l,r]=$ interval of $B$ 
   \IF { $\phi_c(G^{[l,r]})<\phi^*$} 
     \STATE $(C,t,t')=$\emph{Refine($B$)}
     \STATE $\phi^*=\min(\phi^*,\phi(C,t,t'))$
     \STATE Add $(C,t,t')$ to $\mathcal{C}$
   \ENDIF
  \ENDFOR
  \STATE RETURN $\mathcal{C}$
\end{algorithmic}
\label{alg}
\end{algorithm}

The steps of our overall method \ouralg are detailed in Alg. 1. We first pre-compute the eigenvalues for a set $\Phi$ of $O(T)$ intervals as outlined in Sec.~\ref{sec:prune} (Step 1) and find an estimate $\phi^*$ of the solution by probing a constant number of promising periods of small $\lambda_2$ in $\Phi$ (Step 2). We employ a light-weight version of hashing in those periods. Then we prune groups by composing their bounds $\uline{\phi_c}(G^\tau)$ based on $\Phi$ (Step 3), and for unpruned groups we compute composite bounds of individual intervals and attempt to prune them (Step 4). We next hash neighborhoods $N_u^t$ of nodes, targeting all possible scales $s$ for collision that include unpruned intervals (Steps 5-11). To target a particular time scale $s$, we select the appropriate number of time pivots $k^*(s)$ according to Thm.~\ref{thm:optk}(Step 8). Next, we process collision buckets $B$ containing sets of neighborhoods $N_u^t$ ordered by a decreasing fill-factor, quantifying the consistency of node sets in the timestamps with the bucket (Steps 12-19). If the interval spanned by the current bucket's timestamps cannot be pruned, we form the aggregated graph $G^{[l,r]}$ and we compute the lowest temporal community $C$ around the seed nodes in the bucket using the spectral sweep method by Anderesen et al.~\cite{Andersen2006} (Step 15); briefly, this operates by considering just the top-ranking node in the bucket, then the top two, then the top three, and so on. We maintain the best estimate in $\phi^*$ to enable more pruning of buckets to be processed (Step 16) and add $C$ to the result set $\mathcal{C}$ (Step 17). Finally, we report $\mathcal{C}$. Note that we can easily maintain and report multiple top communities in Steps 12-20.
\begin{table*}
\centering
\footnotesize
\begin{tabular}{|l|c|c|c|c|c|c|c|c|c|c|c|}
    \multicolumn{4}{c}{} & \multicolumn{2}{|c|}{PHASR} 
                         & \multicolumn{2}{|c|}{EXH~\cite{Andersen2006}$^*$} 
                         & \multicolumn{2}{|c|}{H+RW~\cite{macropol2010scalable}$^*$} &
    \multicolumn{2}{|c|}{L-metric~\cite{takaffoli2013incremental}} \\
    \hline
    Dataset & $|V|$ & $\bar{|E|}$ & $T$ & Time & $\phi$ & Time & $\phi$& Time & $\phi$& Time & $\phi$ \\ \hline \hline
    Synth. & 1k-15k & 20k-300k & $1k$ & \bf{$76s$} & {\bf $0.017$} & days & n/a & days & n/a & $6486s$ & $0.042$\\ \hline
    Road & $100$ & $128$ & $1k$ & $118s$ & {\bf $0.039$} & days & n/a & $>24h$ & n/a & \bf{$1s$} & $0.099$\\ \hline
    Internet & $2542$ & $12699$ & $120$ & \bf{$2.3h$} & \bf{$0.008$} & days & n/a & $>24h$ & n/a  & $>6h$ & n/a \\ \hline
    Call & $1333$ & $756k$ & $24$ & \bf{$28m$} & 0.0032 & days & n/a & $>24h$ & n/a & $3.5h$ & 0.215 \\ \hline
\end{tabular}
\caption{Data sets used for experimentation and comparison to the L-Metric dynamic community method~\cite{takaffoli2013incremental}. $^*$Comparisons to other baselines were infeasible on the full datasets. See Fig.~\ref{fig:syntheticScal} for results on smaller versions of these datasets.}\vsb\vsa\vsc
\label{tbl:data}\vsb
\end{table*}

\noindent {\em Complexity analysis:\em} Precomputing $\Phi$ requires $O(T)$ eigenvalue computations, since we consider non-overlapping intervals of exponentially increasing sizes. The group pruning requires $O(T\log^2 T |V|)$ time, since when ensuring overlap of at least $\beta<1$ among interval group members in the grouping, we get $O(\log T)$ groups for every starting position and a total of $O(T\log T)$ groups. To compute the composite group bound we need at most $O(\log T)$ precomputed eigenvalue intervals and a scan over the node volumes, arriving at the final group pruning complexity. It is important to note that the eigenvalue computations, pruning and candidate verification can all be trivially parallelized on common MapReduce-like systems. The time spent in the remainder of the algorithm depends on the effectiveness of the pruning steps which, as we show in the evaluation, are able to filter most intervals given the existence of outstanding local temporal communities.  
\vsb
\section{Experimental evaluation}
\vsa
\label{sec:exp}
We evaluate the quality and scalability of our approach in both synthetic and real-world networks.
Of main interest are the running time savings due to the pruning enabled by our bounds and the quality of candidate communities produced by hashing. We conduct all experiments on a $3.6$GHz Intel processor with $16$GB of RAM. All algorithms are implemented as single-thread Java programs.
To compute eigenvalues, we employ the implementation of the Lanczos algorithm from the \emph{Matrix Toolkit Java}\footnote{Matrix Toolkit Java from \url{https://github.com/fommil/matrix-toolkits-java.}} 
\begin{figure*}[ht]
\centering
\hspace{-0.1in}
\subfigure[][Pruning \% (Synthetic)]
{
 \centering
  \includegraphics[width=0.24\textwidth]{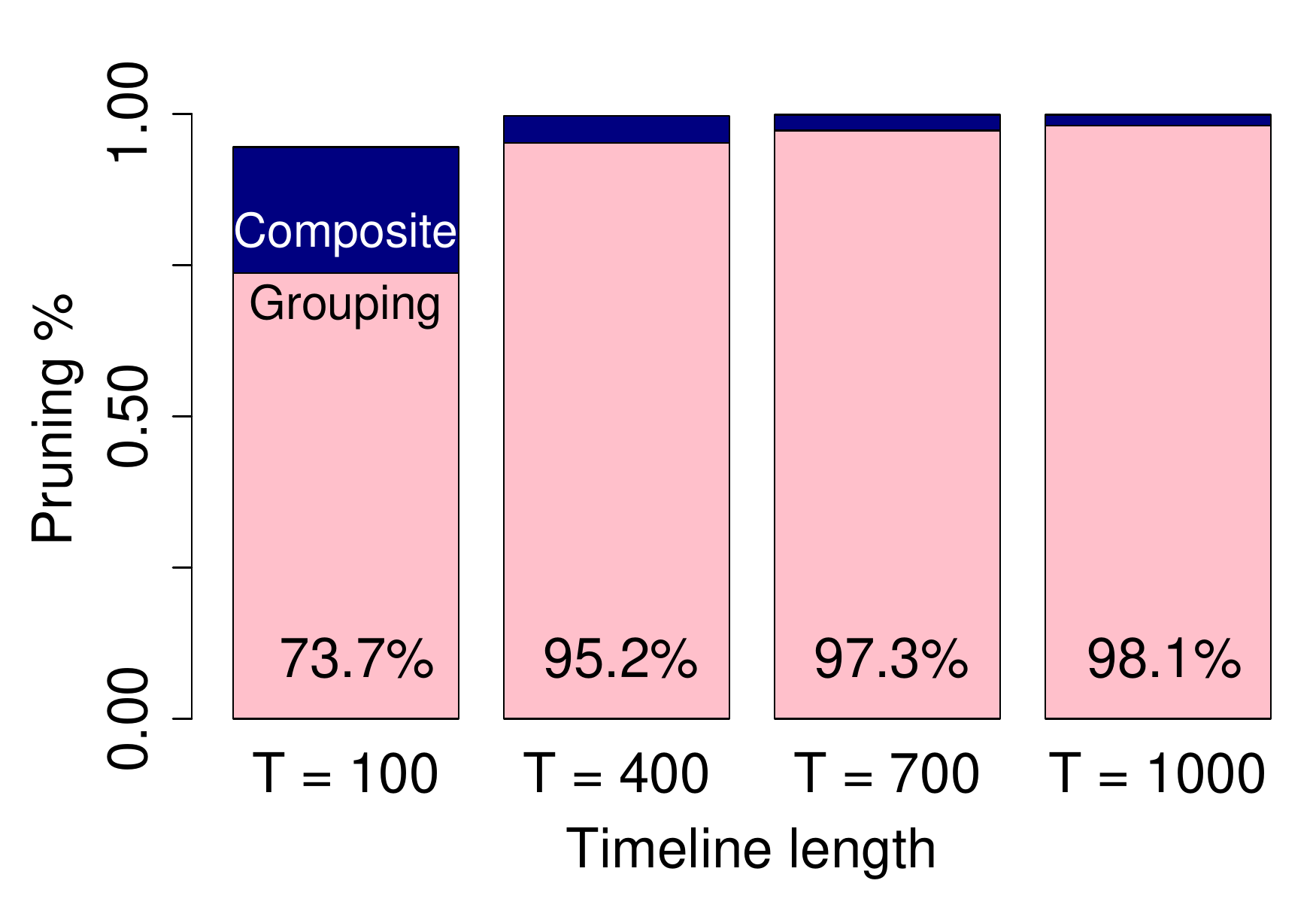}
  \label{fig:pruning_syn}
}\hspace{-0.1in}
\subfigure[][Pruning time (Synthetic)]
{
\centering
  \includegraphics[width=0.22\textwidth]{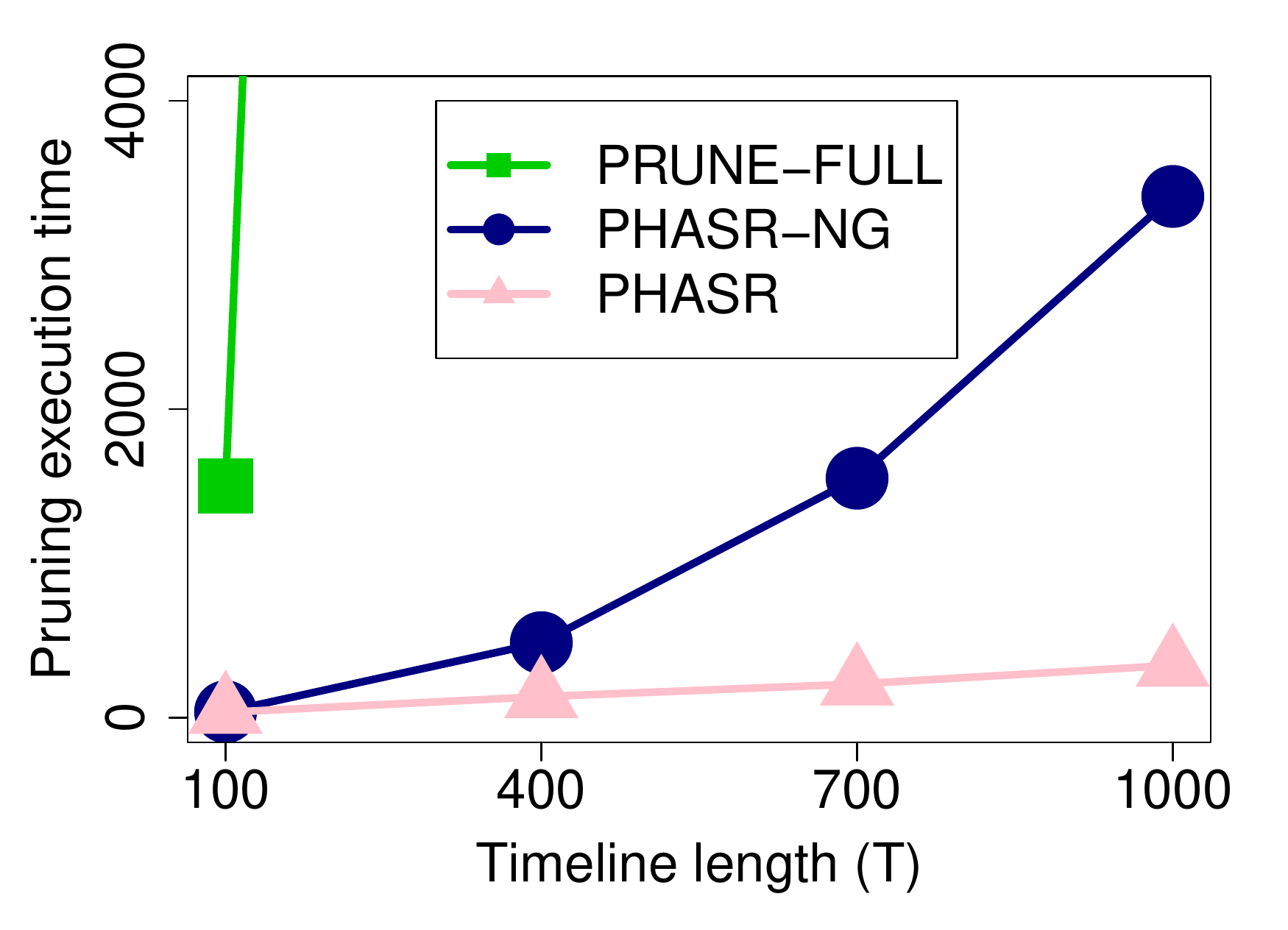}
  \label{fig:time_syn}
}\hspace{-0.1in}
\subfigure[][Pruning \% (Road)]{
\centering\hspace{-0.1in}
  \includegraphics[width = 0.21\textwidth]{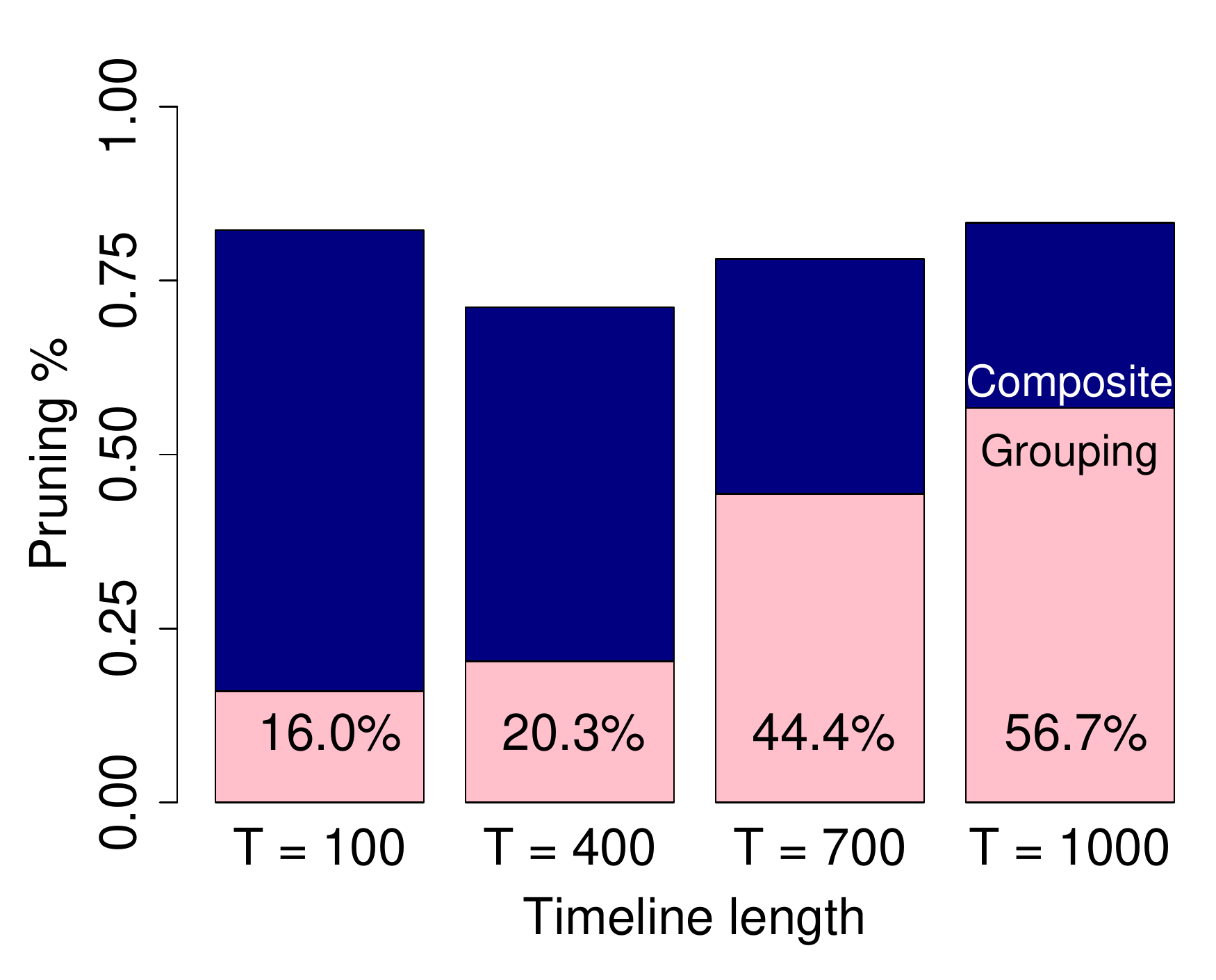}
  \label{fig:pruning_road}
}\hspace{-0.1in}
\subfigure[][Pruning time (Road)]
{
\centering
  \includegraphics[width=0.22\textwidth]{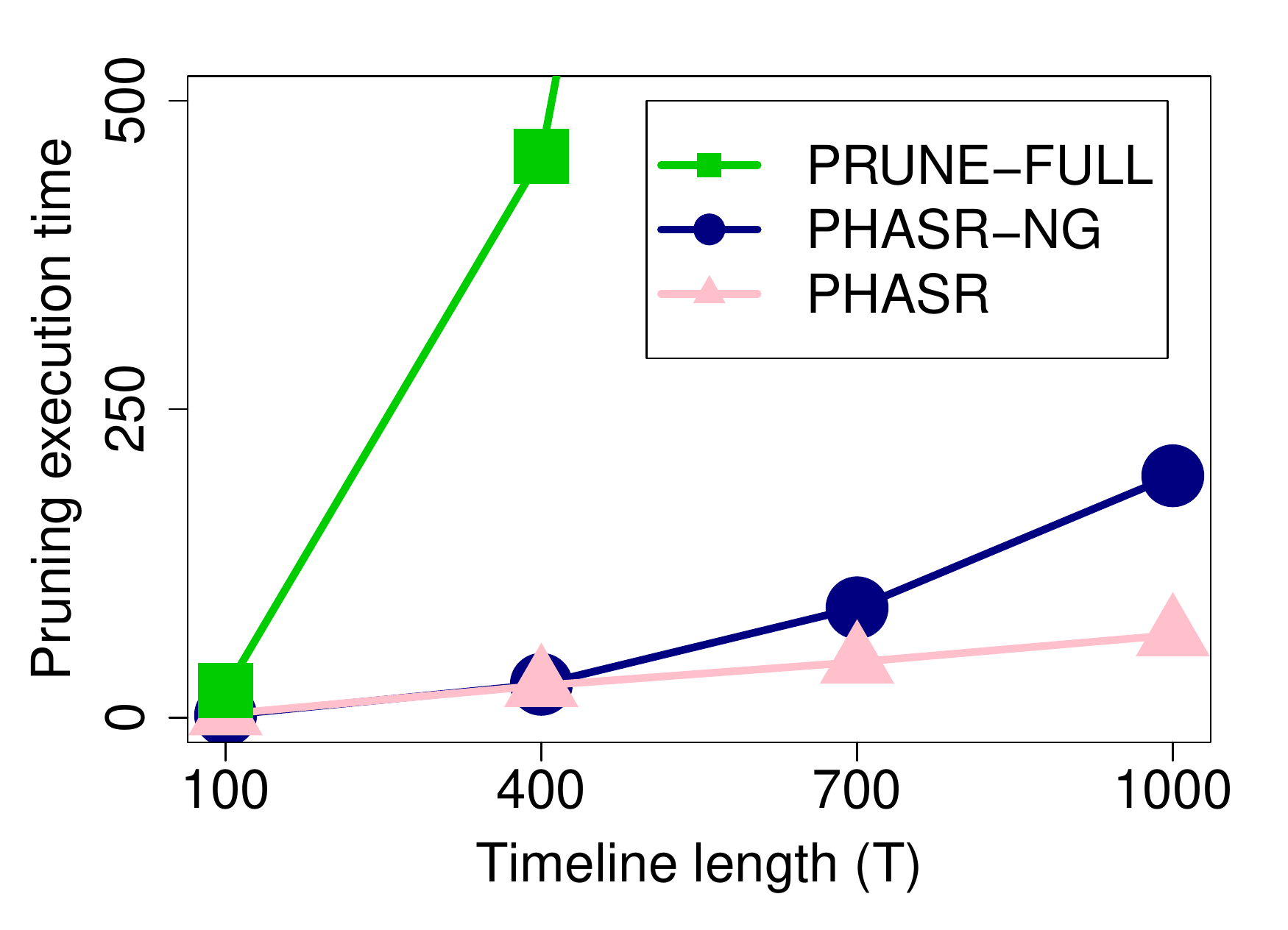}
  \label{fig:time_road}
}\hspace{-0.1in}

\subfigure[][Pruning \% (Internet)]{
\centering\hspace{-0.1in}
  \includegraphics[width = 0.21\textwidth]{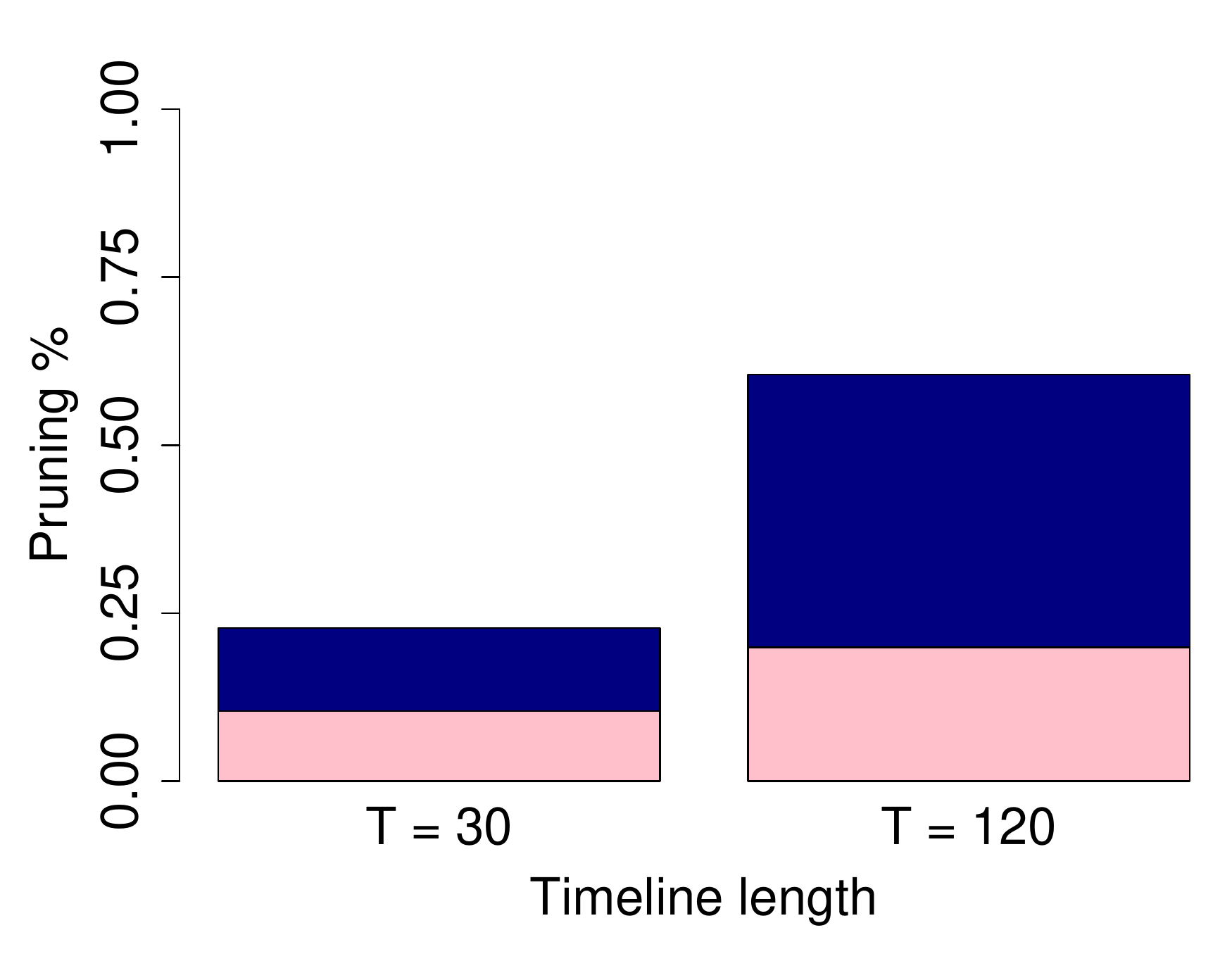}
  \label{fig:pruning_net}
}\hspace{-0.1in}
\subfigure[][Pruning time (Internet)]
{
\centering
  \includegraphics[width=0.21\textwidth]{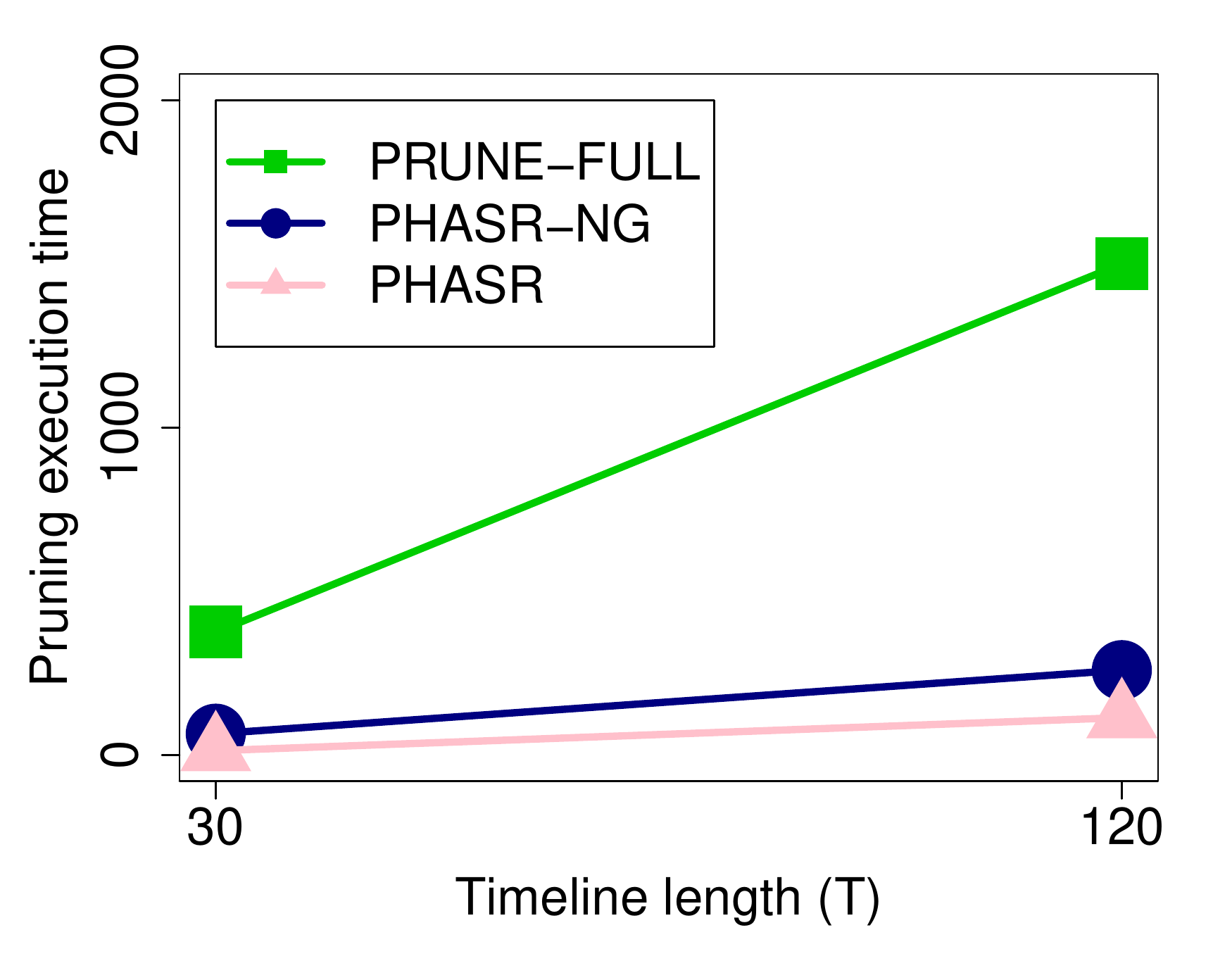}
  \label{fig:time_net}
}\hspace{-0.1in}
\subfigure[][Effect of $\eta(\alpha)$]
{
 \centering
  \includegraphics[width=0.23\textwidth]{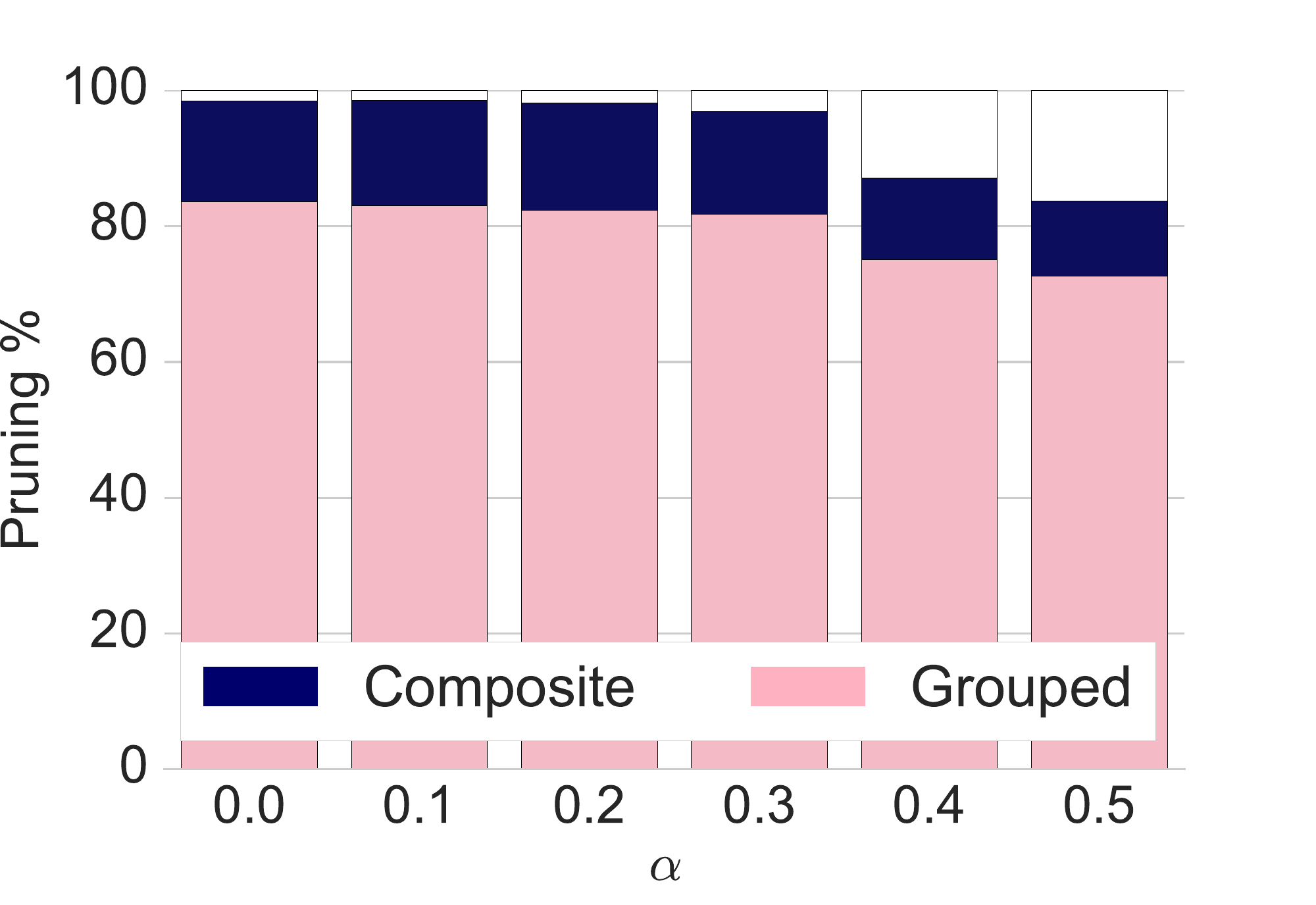}
  \label{fig:prunenorm}
}\hspace{-0.2in}
\subfigure[][Heatmap of pruned intervals]{
\centering
  \includegraphics[width = 0.24\textwidth]{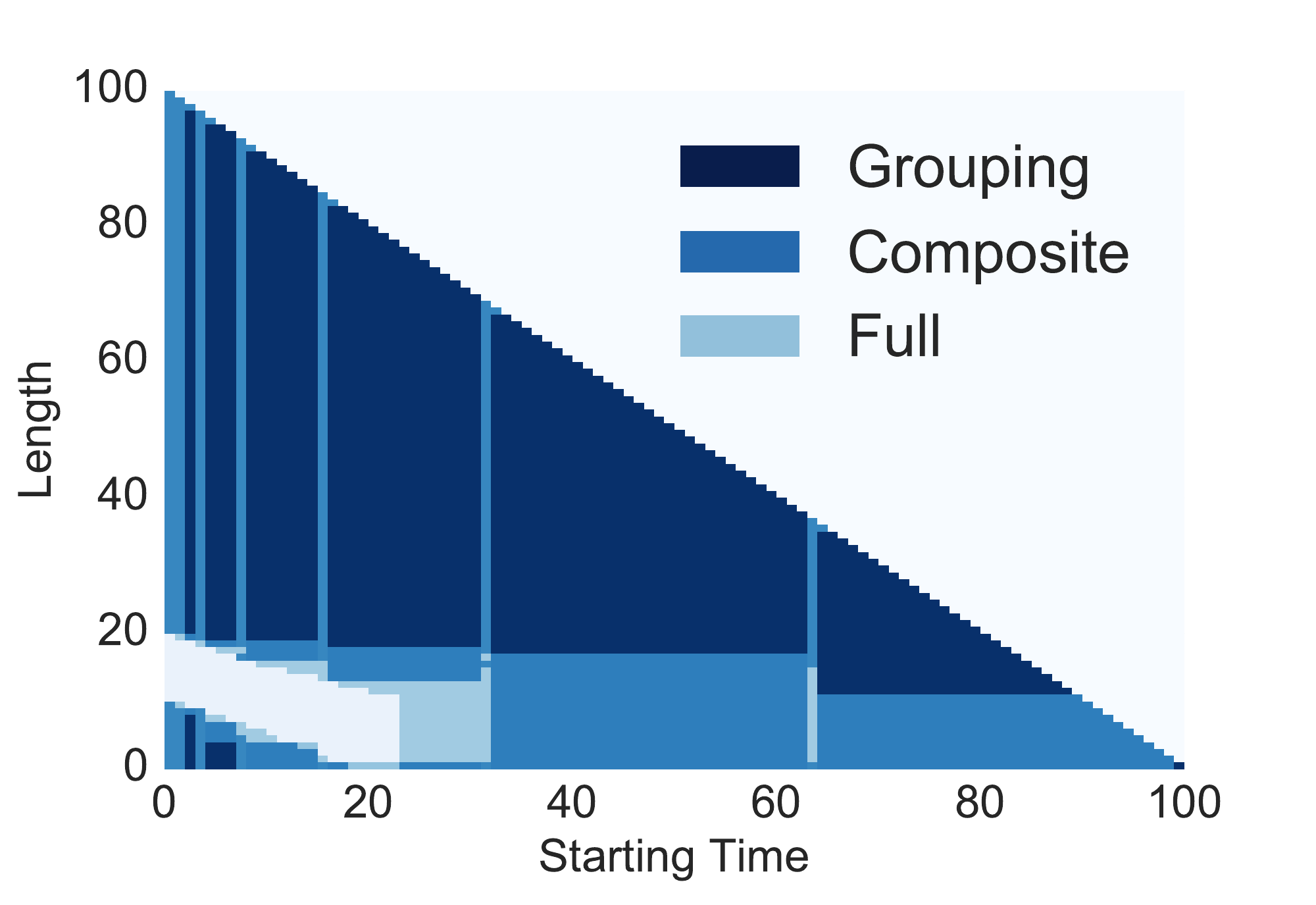}
  \label{fig:heat}
}\hspace{-0.2in}\vsb

\caption{\subref{fig:pruning_syn} Percentage of all pruned intervals using the group and then composite pruning in \ouralg. \subref{fig:time_syn} Comparison of the time taken for pruning with (\ouralg) and without grouping (\ouralg-NG), and full calculation of all Cheeger bounds (PRUNE-FULL). Similar comparisons for the Internet~\subref{fig:pruning_net},\subref{fig:time_net} and Road~\subref{fig:pruning_road},\subref{fig:time_road} networks. \subref{fig:prunenorm} Pruning power in synthetic data for different levels of temporal normalization $\alpha$. \subref{fig:heat} Heatmap representation of the pruned intervals in one of the synthetic datasets for $T=100$ where each pixel represents an interval with horizontal coordinate its starting point and vertical coordinate its duration. Colors encode the kind of bound that enabled pruning the specific interval.}
\label{fig:PruneRes} \vsc
\end{figure*}
\vsb
\subsection {Datasets and competing techniques}

\noindent{\bf Datasets: } We use preferential attachment synthetic networks~\cite{barabasi-albert} of sizes between $1k$ and $15k$ nodes and average degree of $20$. 
We replicate the unweighted network structure over $T=1000$ timestamps and assign Poisson random weights with mean $5$ on edges in time independently. We inject a strong temporal community $(C,t,t')$ of length $t'-t=10$ by increasing the average weight on the internal edges. The community \emph{contrast} is defined as the ratio of the mean weight in $(C,t,t')$ and the global average of $5$; we used a contrast value of $8$ in synthetic data unless otherwise specified.

We also use real-world datasets of various length, number of nodes and density listed in Tab.~\ref{tbl:data}. The \emph{Road} traffic dataset is a subnetwork of the California highway system. Edge (road segments) are weighted based on the average speed at $5m$ intervals. In this dataset we aim to detect contiguous subnetworks of abnormal speeds over time. To detect high- and low-speed temporal subgraphs we assign weights as $\mathcal{V}^2$ or $(85-\mathcal{V})^2$ respectively, where $\mathcal{V}$ is the speed in mph at a given time. Execution times for both weighting schemes are similar. The \emph{Internet} traffic data is a $2h$ trace of all p2p web traffic at the level of organizations (first three bytes of host IPs) from June 2013 on a backbone link in Japan, where weights are assigned as the number of packets between a pair of organizations at $1m$ resolution~\cite{Cho2000mawi}. Our densest dataset is a \emph{Call} graph among sectors of the city of Milan, Italy over $24h$ period~\cite{OBD}. Edge weights correspond to the number of calls between sectors within an hour. 

\noindent{\bf A note on data sizes and parallelization opportunities: }It is important to note that while the real-world networks we employ for experimentation are in the order of thousands of nodes, the search space over all possible intervals requires consideration of $O(T^2)$ differently weighted graphs of that size. For example, in our \emph{Call} dataset all-interval graphs contain cumulatively \textit{208 million edges} while in the largest Synthetic dataset the cumulative number edges exceeds \textit{1.4 billion}. This large search space is the reason why exhaustive baselines do not scale to instances of such sizes. In addition, our technique can be easily parallelized employing a bulk synchronous processing system (BSP) such as Hadoop~\footnote{Hadoop. \url{http://hadoop.apache.org/}}, since the underlying building blocks of eignevalue computations, hashing and aggregation of colliding neighborhoods fit naturally the BSP programming paradigm. A parallel implementation and corresponding scalability experiments are beyond the scope of this work, but we plan to included them in an extended version of this work.

\noindent{\bf Baselines: }We compare \ouralg to three baselines: (1) exhaustive (\emph{EXH}) temporal extension of the spectral sweep method by Andersen et al.~\cite{Andersen2006}; (2) a temporal extension of the hashing community detection (\emph{H+RW})~\cite{macropol2010scalable}; and (3) the incremental \emph{L-Metric} for local communities in dynamic graphs by Takaffoli et al.\cite{takaffoli2013incremental}. \emph{EXH} performs spectral sweeps in all possible intervals and starting from all nodes. \emph{H+RW} hashes neighborhoods in the graphs of all possible time intervals (thus can be viewed as naive dynamic extension of~\cite{macropol2010scalable}), and then performs a sweep from seeds identified by hashing. Hashing candidates do not form low-conductance communities on their own since~\cite{macropol2010scalable} does not consider cuts. \emph{L-Metric}~\cite{takaffoli2013incremental} incrementally extends local communities in time, by using the connected components of communities from the previous time step as seeds. 
Since it allows communities to change over time, we implement a post-processing step in which we maintain the largest node intersections for all possible intervals of a contiguous community in order to obtain dynamic communities of fixed membership. We also consider two versions of \ouralg: explicit computation of all interval bounds for pruning, termed \emph{(PRUNE-FULL)} and our method using composite and group bounds \ouralg. We evaluate the pruning power for different pruning strategies, the scalability of competing techniques and the effect of parameters for \ouralg in what follows.

\vsb
\subsection{Pruning power}
\vsa

\begin{figure*}[t]\vsb
\centering
\hspace{-0.1in}
\subfigure[][Scalability Synthetic]
{
 \centering
  \includegraphics[width=0.238\textwidth]{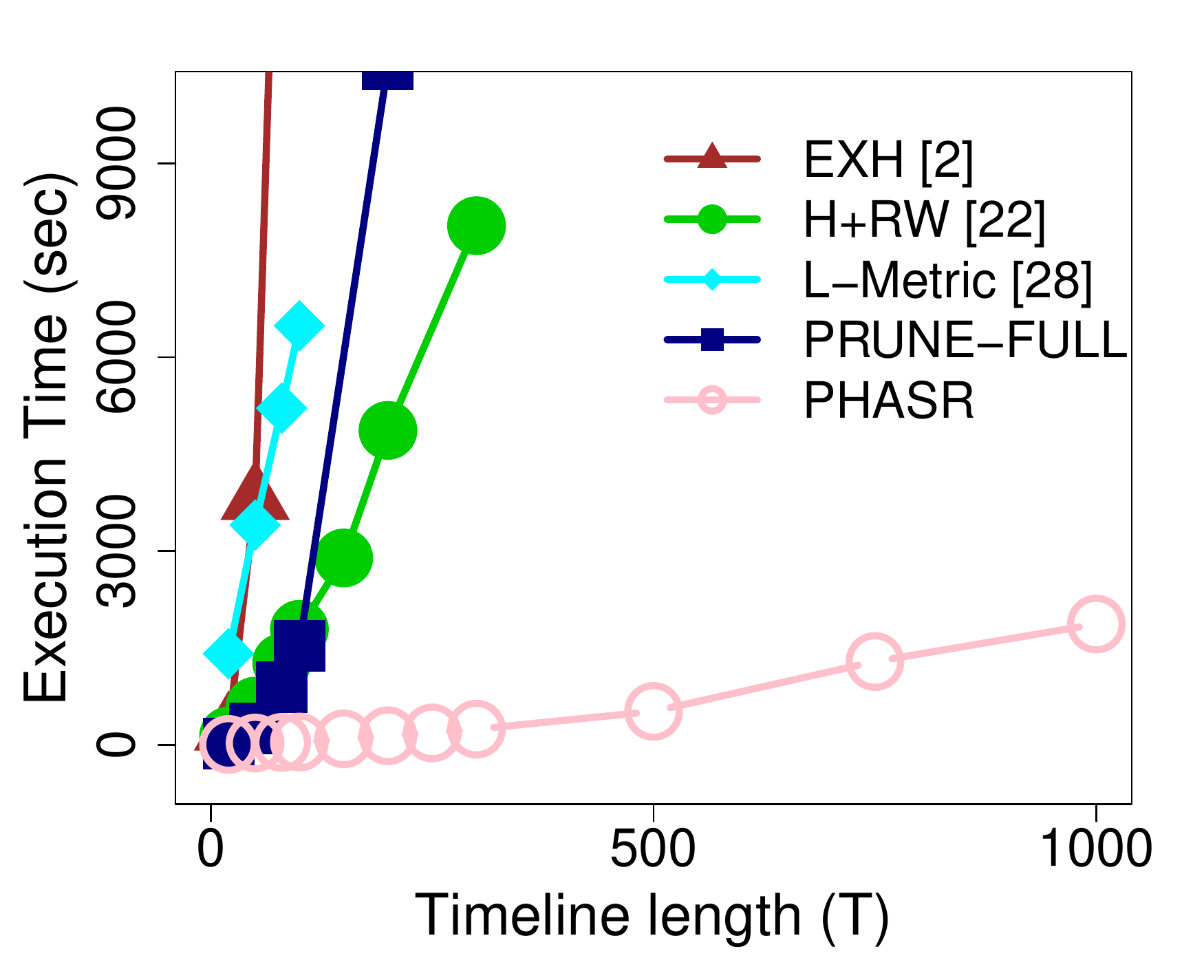}
  \label{fig:ssyn}
}
\hspace{-0.1in}
\subfigure[][Scalability Road]{
\centering
  \includegraphics[width = 0.25\textwidth]{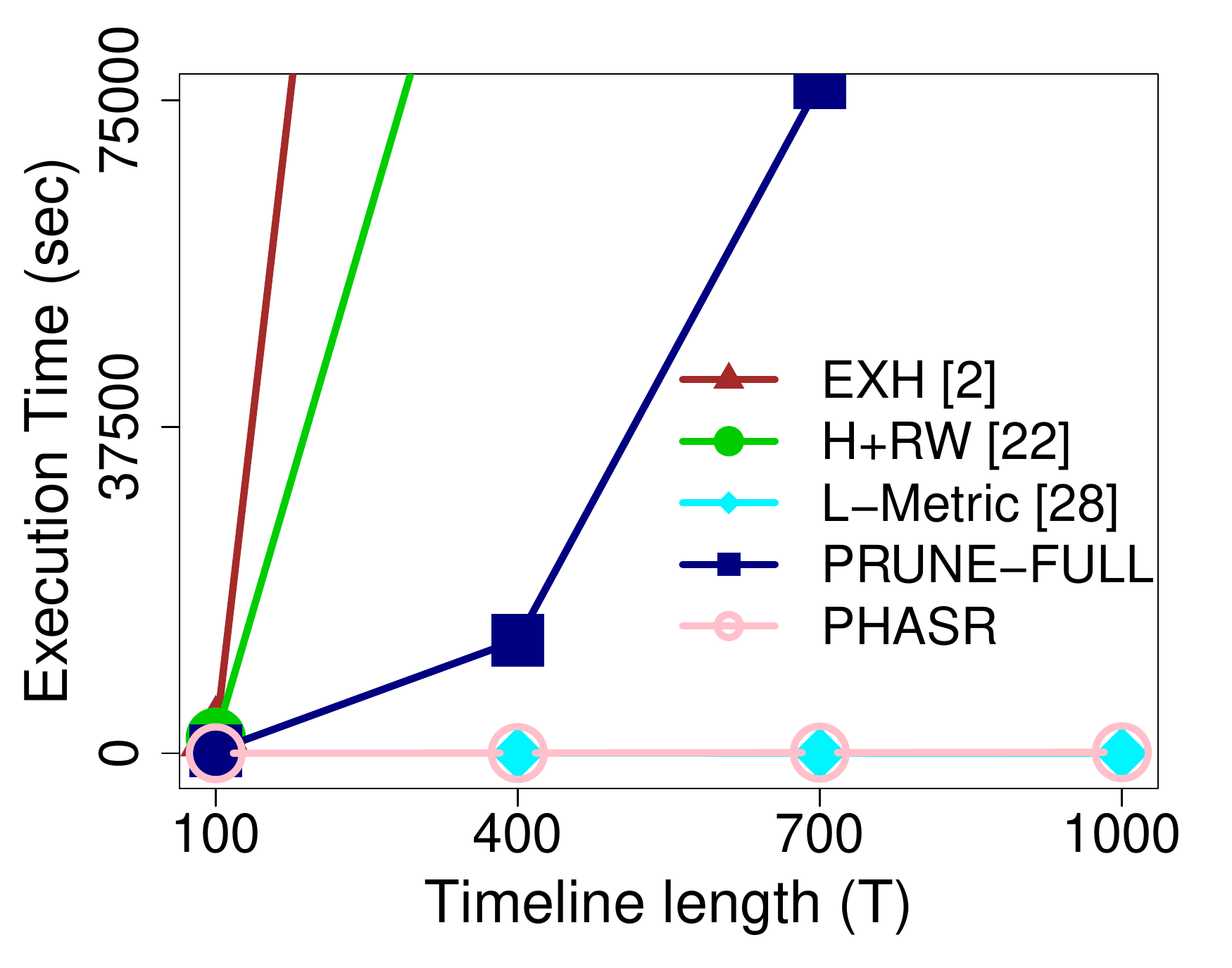}
  \label{fig:scaleRoad}
}\hspace{-0.1in}
\subfigure[][Scalability Internet]
{
\centering
  \includegraphics[width=0.25\textwidth]{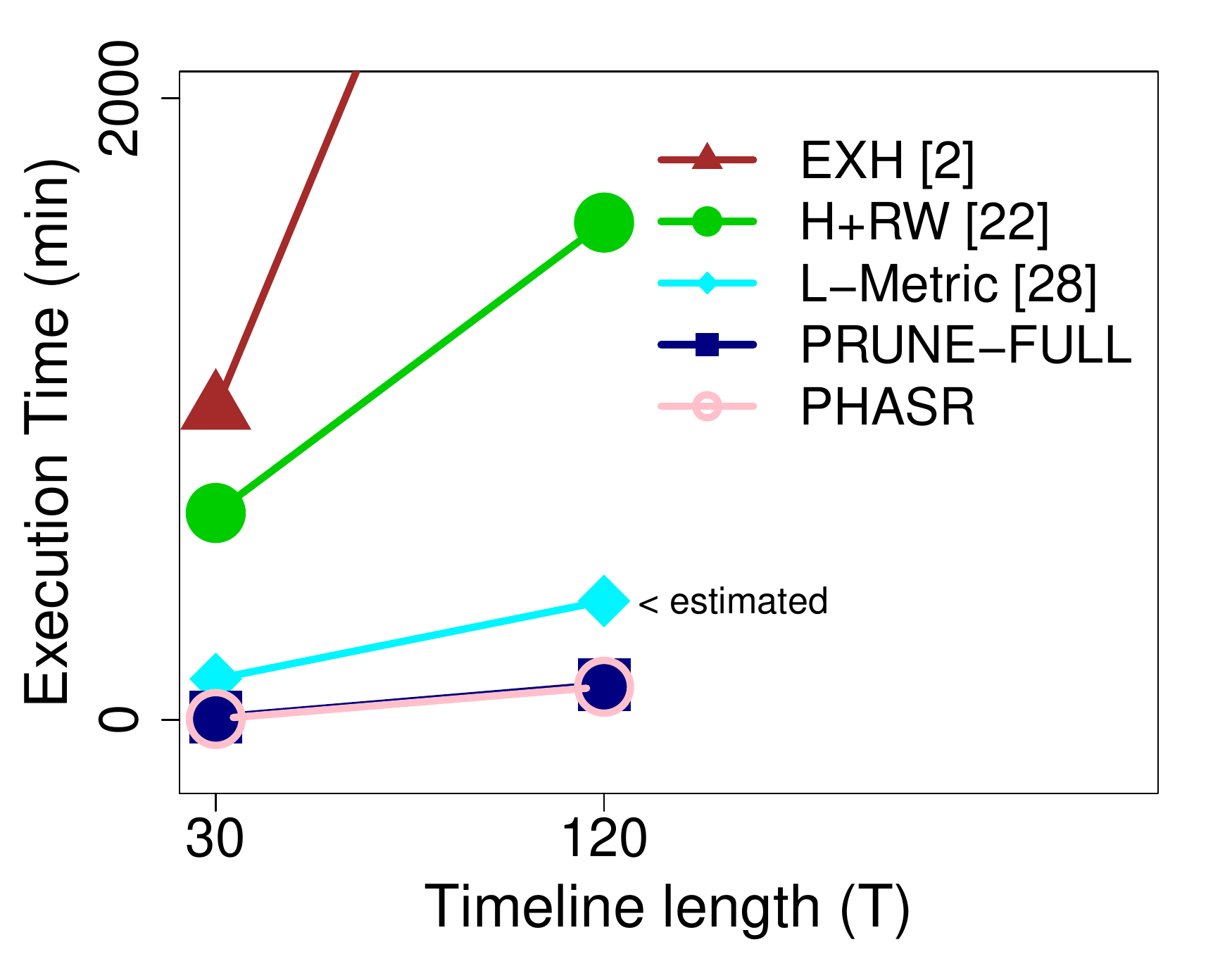}
  \label{fig:scaleInternet}
}
\hspace{-0.1in}
\subfigure[][Scalability in $|V|$]{
\centering\hspace{-0.1in}
  \includegraphics[width = 0.24\textwidth]{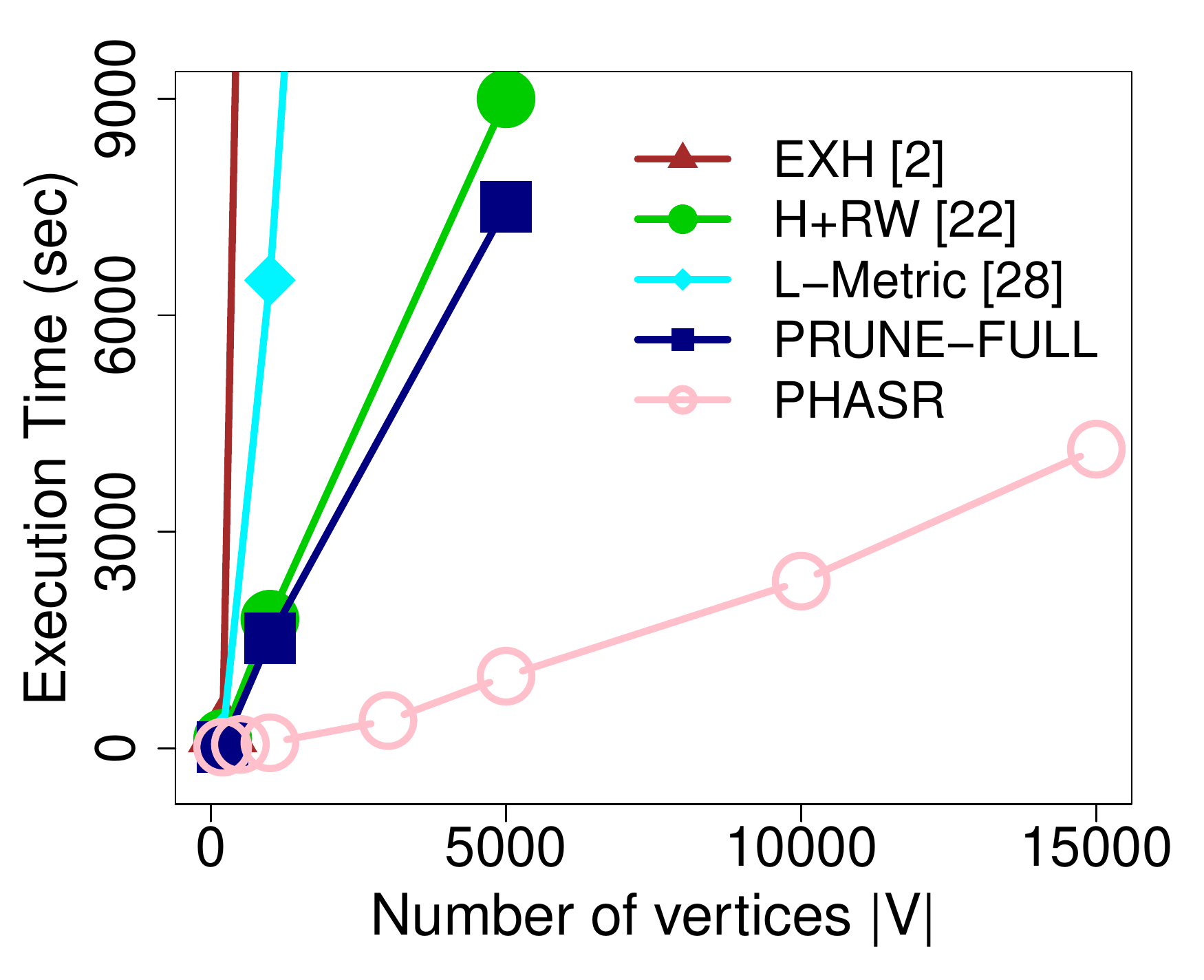}
  \label{fig:scalegraph}
}\hspace{-0.1in}\vsb
\caption{Comparison of \ouralg's total running time with that of alternatives for increasing $T$ in \subref{fig:ssyn} Synthetic ($|V|=1000$); \subref{fig:scaleRoad} Road; and \subref{fig:scaleInternet} the Internet datasets. \subref{fig:scalegraph}: Scalability with the graph size $(T=100)$ on the Synthetic dataset.  
}
\label{fig:syntheticScal}\vsb
\end{figure*}
We evaluate the pruning power of our bounds in both synthetic and real world datasets Fig.~\ref{fig:pruning_syn} shows the average pruning percentage of all intervals in the synthetic data for increasing $T$ and an injected community of a fixed length of $10$. We prune more than $95\%$ of the possible intervals across all lengths. The fast grouping phase prunes the majority of the intervals ranging from $73\%$ when the injected community is $1/10$-th of the timeline to more than $98\%$ of the intervals for $T=1000$. The execution time is presented in Fig.~\ref{fig:time_syn}. Since most intervals are pruned at the group stage, \ouralg's pruning time grows almost linearly with $T$, while the time for pruning based on composite bounds without grouping (\ouralg-NG) grows faster, resulting in more than an order of magnitude savings at $T=1000$ due to grouping. Computing all-interval eigenvalues (Full) does not scale beyond $100$ time steps, requiring two orders of magnitude more time then composite and group pruning due to the expensive eigenvalue computations. 

We perform similar pruning evaluation for the Road Figs.~\ref{fig:pruning_road}, \ref{fig:time_road} and Internet Figs.~\ref{fig:pruning_net},\ref{fig:time_net} datasets. We prune more than $75\%$ of the intervals in Road and $55\%$ in Internet. Grouping is less effective here as there are multiple temporal intervals with good temporal conductance, though still providing increased effectiveness reflected in the widening gap between the running times of \ouralg and \ouralg-NG for Road Fig.\ref{fig:time_road}. The savings of grouping are smaller for Internet as there exist communities of hosts of low conductance persisting over most of the $2h$ span. Exhaustive computation of eigenvalues, in comparison is at least an order of magnitude slower for the longest timespans of both datasets. The pruning percentage in the Call dataset is more than $75\%$ when considering all $24$ time intervals, allowing for fast overall time regardless of its density (no figure due to the short $T=24$). 

We also study the effect of temporal normalization on pruning effectiveness. The temporal normalization $\eta(t,t')=(t'-t)^{-\alpha}$ is controlled by the exponent parameter $\alpha$, where higher values of $\alpha$ decrease the temporal conductance of longer intervals and thus make them more preferable. As we discussed earlier $\alpha=0$ amounts to no normalization and in this case solutions tend to reside in a single timestamp. Alternatively, if $\alpha$ is very large solutions that span the whole timeline are preferred, though for such settings, simply aggregating the graph over all timestamps and employing static graph solutions will work better. We show the effect of targeting medium interval-length solutions (the most challenging setting) on the effectiveness of our pruning in Fig.~\ref{fig:prunenorm}. The total pruning decreases from $98\%$ to $83\%$ for values of $\alpha$ up to $0.5$ and the fraction of both intervals pruned by the group bound (blue) and composite bound (red) reduce proportionally. Nevertheless, this level of pruning ensures significantly more efficient overall processing than employing hashing for all times and temporal scales in an exhaustive manner. 

A detailed visualization of the pruning effectiveness in Synthetic is presented as a heatmap in Fig.~\ref{fig:heat}. The space of all possible intervals is represented in a lower triangular matrix of pixels, where the pixel position encodes an interval start time (horizontal) and length (vertical axis). Grouping (darkest shade) prunes most of the long intervals and the majority of intervals of size less than $20$ are pruned by the composite interval bound. Employing the Full eigenvalue computation may prune only a small percentage of additional intervals at the cost of lengthy eigenvalue computations (lightest shade). Intervals of significant overlap with the injected community (times $20$-$30$) remain unpruned and considered for hashing.  

\vsb
\subsection{Scalability.}
\ouralg scales well with increasing timeline length $T$. Particularly, its pruning phase is very efficient as evident in Figs.~\ref{fig:time_syn},\ref{fig:time_net},\ref{fig:time_road}. The na{\"i}ve approach of directly calculating all possible Cheeger bounds is quadratic and quickly becomes infeasible. The use of pruning groups here is key in Synthetic, Road and Call---while composite bounding without groups (\ouralg-NG) is much better than the na{\"i}ve approach, only the full algorithm with pruning groups produces the near-linear scaling that is necessary for very long timelines. 

The last $4$ columns of Tab.~1 show the total running time and conductance of the best solutions for \ouralg and L-Metric for the full datasets. In all cases except the Road dataset, \ouralg completes much faster than the L-Metric: $10x$, $>6x$ and $7x$ faster for Synthetic ($|V|=1k$), Internet and Call respectively. L-Metric does not complete in more than $6$ hours on Internet ($T=120$). In all cases in which $L-Metric$ completes, the discovered communities are of significantly worse conductance: $2.5$ times worse in Synthetic and Road, and $67$ times worse in Call. The reason for this lower quality is that L-Metric considers only adjacent time steps when trying to reconcile communities in time, while \ouralg considers the full evolution of the graph at different scales.

Fig.~\ref{fig:syntheticScal} shows the complete running time of \ouralg algorithm---pruning, hashing, and refinement via random walks---versus competing techniques over increasing $T$ (Figs.~\ref{fig:ssyn},\ref{fig:scaleRoad},\ref{fig:scaleInternet}). Large numbers of vertices or time periods quickly render the exhaustive competitor methods infeasible, with runtimes of many hours or even days (estimated). The pruning and hashing segments of our approach scale well in both time and graph size; for large $T$, only about two percent of total execution time is spent on pruning, and the remainder on hashing and refinement. It is important to note that the running time of hashing and refinement can be reduced by considering smaller number of hash functions and bands at the expense of possibly worse-conductance results. A faster push-based local RW estimation, as described in~\cite{Andersen2006}, will enable scaling to larger network sizes as well. Our current implementation features only a naive full-network RWR, since scalable implementation of spectral sweeps is not the main focus of this work. L-Metric's running time grows quickly with $T$ and the graph size (Fig.~\ref{fig:scalegraph}) and is dominated by our approach on all datasets, but the small Road dataset in which it completes faster, but discovers a worse community (Tab.~1). \ouralg equipped with group pruning is the only alternative that scales with the graph size (Fig.~\ref{fig:scalegraph}) and can be further improved by trivial parallel implementation as discussed earlier. 

\subsection{Case studies: {\em Call\em} and {\em Internet\em}}

We obtained several low-conductance communities in the Milan telecom ({\em Call\em}) data; their locations are shown in Fig.~\ref{fig:Mmap}. Multiple small communities formed at various times during the day in the circled region; their size suggest they may simply be coincidental. However, a larger, more coherent community was discovered in the area marked with a rectangle. Because it appears near the A50 highway during the early morning hours, we speculate it may be the result of calls from vehicles in commuter traffic.

The dominant community in the Internet traffic data covered much of the timeline we examined. The nodes involved all represented major telecom companies in a variety of countries: Japan, Saudi Arabia, Korea, Israel, and the United Kingdom; the low conductance here appears to be because of extremely high traffic rates between what we speculate are backbone Internet providers.

\vsb
\subsection{Effect of parameters}
\vsa
We also evaluate the effect of hashing parameters on the quality of obtained seeds. Our experiment demonstrate that higher number of hashing bands and neighborhood hashing functions increases the quality of obtained seeds; however, we observe diminishing returns past $7$ bands in Synthetic. Details omitted due to space limitations.

\vsb
\section{Conclusions}
\vsa
We proposed the problem of local temporal communities with the goal of detecting a subset of nodes and a time interval in which the nodes interact exclusively with each other. We generalized the measure of conductance to the temporal context and proposed a method \ouralg for the minimum conductance temporal community. To scale the search in time we employed a novel spectral pruning approach that is sub-quadratic in the length of the total timeline. To scale the search in the graph space we proposed a time-and-graph locality sensitive family for neighborhoods of nodes which effectively spots cores of good communities in time. 

We evaluated \ouralg on both real and synthetic datasets and demonstrated that it scales better than alternatives to large instances, achieving two orders of magnitude running time reduction  in Synthetic and 3 to 7 times reduction on big real instances compared to alternatives. \ouralg also discovered communities of $2$ to $67$ times lower conductance than those obtained by a dynamic community baseline from the literature. This performance and accuracy is enabled by pruning as much as $95\%$ of the possible time intervals in Synthetic and between $55\%$ and $75\%$ in real-world datasets; and due to our effective temporal hashing scheme for spotting good seeds in unpruned intervals. 

\vsb
\bibliographystyle{plainnat}
\bibliography{ref}
\vsb
\section*{Appendix: Proofs of theoretical results.}

\noindent{\bf Proof of Theorem~\ref{thm:tlocal}: Temporal locality.}

\begin{proof}\vsb
Consider two timestamps, $t \leq t'$. Their hash values with respect to a $k$-partitioning $\tau^k(t)$ and $\tau^k(t')$ will differ only when $\exists p_i | t < p_i \leq t'$. Since every pivot $p_i$ was chosen uniformly at random from the full timeline, the probability that it falls between $t$ and $t'$ is proportional to the duration between the timepoints divided by the timeline length: $\frac{t' - t}{T}$. The probability that a specific pivot is not selected between the two timepoints is $1-\frac{t' - t}{T}$. The two hash values match if none of the $k$ pivots is selected between them, and the probability of this event is $(1-\frac{t' - t}{T})^k$ since all $p_i$ are chosen independently. When $t'-t \leq \Delta_1$ this probability is greater than $(1-\frac{\Delta_1}{T})^k$, and the analogous relationship holds for $t'-t \geq \Delta_2$.\vsb
\end{proof}

\noindent{\bf Proof of Theorem~\ref{thm:optk}: Optimal Number of pivots.}
\begin{proof}\vsb
Observe that the probability of a pivot landing in a perfect bracketing position (on the left or the right) is $\frac{1}{T}$. The probability of landing anywhere outside our target community is $1 - \frac{\Delta^*}{T}$. Any of the $\binom{k}{2}$ pairs of pivots could be the "bookends", and these could occur in either order (left-right or right-left). Thus, the probability of a perfect partition is: $$2\binom{k}{2}\big(\frac{1}{T}\big)^2\big(1 - \frac{\Delta^*}{T}\big)^{k-2} =
\big(\frac{1}{T}\big)^2\big(1 - \frac{\Delta^*}{T}\big)^{k-2}\big(k^2 - k\big)\mathrm{.}\vsb$$


Since we want to find the $k$ that maximizes the probability of perfect partition, we find $\frac{\partial p}{\partial k}$:
$$\big(1-\frac{\Delta^*}{T}\big)^{k-1}\big[(2k-1)(1-\frac{\Delta^*}{T})^{k-1}+
(k^2-k)\log(1-\frac{\Delta^*}{T})\big]$$
Since $0 < \frac{\Delta^*}{T} < 1$, one of the roots of $\frac{\partial p}{\partial k}=0$ is strictly smaller than $1$, and the only feasible solution is:
$$ k^* = \frac{\log(1-\frac{\Delta^*}{T}) - 2 - \sqrt{\log^2(1-\frac{\Delta^*}{T}) + 4}}{2 \log(1-\frac{\Delta^*}{T})}\mathrm{.}$$
This is well approximated by $\frac{2T}{\Delta^*}$ for $\frac{2T}{\Delta^*}\in[0,1]$.\end{proof}

\noindent{\bf Proof of Lemma~\ref{lem:dom}.}

\begin{proof}\vsb
We have a convex minimization function over a convex set: $\min f^TAf$, such that $f\perp (d+\epsilon)$. The Lagrange form is:
$L(f,\lambda)=f^TAf+\lambda f^T(d+\epsilon).$
Setting the gradient w.r.t. $f$ to zero gives us:
$\frac{\partial L}{\partial f} = 2Af + \lambda (d+\epsilon) = 0 \Rightarrow
f^*=-\frac{\lambda}{2}A^{-1}(d+\epsilon).$
Note that $A^{-1}$ exists since $A$ is positive semi-definite. By an analogous development, $g^*=-\frac{\lambda}{2} A^{-1}d.$ Next we substitute $f^*$ into the objective to obtain the inequality of interest:\vsb

\begin{align*}
\min_{f\perp d+\epsilon}f^TAf &=  {f^*}^TA{f^*} \\
                & =\frac{\lambda^2}{4}(A^{-1}(d+\epsilon))^TA(A^{-1}(d+\epsilon)) \\
                & = \frac{\lambda^2}{4}(d+\epsilon)^TA^{-1}(d+\epsilon)) \\
                & = \frac{\lambda^2}{4}\Big(d^TA^{-1}d + \epsilon^TA^{-1}d + d^TA^{-1}\epsilon + \epsilon^TA^{-1}\epsilon \Big) \\
                & \geq \frac{\lambda^2}{4} d^TA^{-1}d
                 = {g^*}^TAg^* 
                 = \min_{g\perp d} g^TAg
\end{align*} \end{proof}

\noindent{\bf Proof of Theorem~\ref{thm:sumbound}: Composite bound.}

\begin{proof}\vsb
Let $D_i$ denote the degree matrix of $G^{[t_i,t_i']}$ and $\hat{D}$ the degree matrix of $G^{[t,t']}$. According to the \emph{Min-max} theorem the second eigenvalue of $\hat{\mathcal{N}}$ can be characterized as the minimum of the Rayleigh quotient in the subspace orthogonal to the first eigenvector $\hat{d}^{1/2}$~\cite{Spielman2010}:
$\lambda_2(\hat{N}) = \min_{g\perp \hat{d}^{1/2}} \frac{g^T\hat{\mathcal{N}}g}{g^Tg}= \min_{f\perp \hat{d}} \frac{f^T\hat{\mathcal{L}}f}{f^T\hat{D}f},$
where $\hat{d}$ is a vector of the node volumes in $[t,t']$ and $f$ and $g$ are column vectors. The second equality is obtained by variable change $g=\hat{D}^{1/2}f$ and has the form of a generalized eigenvalue problem. We can then show:\vsb
    \begin{align*}\vsc
     \lambda_2(\hat{N}) & = \min_{f\perp \hat{d}} \frac{f^T\hat{\mathcal{L}}f}{f^T\hat{D}f}
                          = \min_{f\perp \hat{d}} \sum_{i=1}^{k}\frac{f^T\mathcal{L}_if}{f^T\hat{D}f} & (1)\\
                        &  = \min_{f\perp \hat{d}} \sum_{i=1}^{k}\frac{f^TD_if}{f^T\hat{D}f}f^T\mathcal{N}_if &(2) \\
                          & \geq \min_{f\perp \hat{d}} \sum_{i=1}^{k}\frac{f^TD_if}{f^T\hat{D}f}\lambda_2(\mathcal{N}_i) &(3)\\
                          & = \min_{g\perp \hat{d}^{1/2}} \sum_{i=1}^{k}\frac{g^T(\hat{D}^{-1/2}D_i\hat{D}^{-1/2})g}{g^Tg}\lambda_2(\mathcal{N}_i) &(4)\\
                            & = \min_{g\perp \hat{d}^{1/2}} \sum_{i=1}^{k}\frac{\sum_{u\in V}g_u^2\frac{vol(u,t_i,t_i')}{vol(u,t,t')}}{\Vert g\Vert^2}\lambda_2(\mathcal{N}_i) &(5)\\
                          & \geq \sum_{i=1}^{k} \min_{u\in V}\frac{vol(u,t_i,t_i')}{vol(u,t,t')}\lambda_2(\mathcal{N}_i) &(6) \vsb
    \end{align*}

For (1) we use the fact that $\hat{\mathcal{L}}=\sum_{i=1}^k\mathcal{L}_i$. In (2) we multiply all summands by $f^TD_if/f^TD_if$. (3) follows from Lem.~\ref{lem:dom} and 
in (4) we have substituted back the vector variables $g=\hat{D}^{1/2}f$. In (5) we have applied the definitions of the sub-interval volume for nodes and have unfolded the quadratic form for the diagonal matrix $g^T(\hat{D}^{-1/2}D_i\hat{D}^{-1/2})g$. Finally, since:
 \begin{align*}\vsb
\sum_{u\in V}g_u^2\frac{vol(u,t_i,t_i')}{vol(u,t,t')} & \geq =\Vert g\Vert^2min_{u\in V}\frac{vol(u,t_i,t_i')}{vol(u,t,t')},\vsb
\end{align*}
\noindent  we obtain the inequality in (6). 
\end{proof}

\noindent{\bf Proof of Theorem~\ref{thm:groupbound}: Group composite bound.}

\begin{proof}\vsb
Let $t^*$ be the end-point of one of the subintervals in the group, i.e. $t\leq t'\leq t^*\leq t''$. Then, since $\eta()$ is monotonically decreasing as long as $\alpha \geq 0$, we have $\eta(t,t^*) \geq \eta(t,t'')$. In addition, since the aggregated weights in super-intervals dominate those in sub-intervals we have that: $\min_{u\in V}\frac{vol(u,t_i,t_i^*)}{vol(u,t,t^*)} \geq \min_{u\in V}\frac{vol(u,t_i,t_i^*)}{vol(u,t,t'')}$. Using the above we have:\vsc\vsc

\begin{align*}\vsc
\uline{\phi_c}(G^{[t,t^*]}) & = \eta(t,t^*)\sum_{i=1}^{k} \min_{u\in V}\frac{vol(u,t_i,t_i^*)}{vol(u,t,t^*)}\lambda_2(\mathcal{N}_i)\\
                            & \geq \eta(t,t'')\sum_{i=1}^{k} \min_{u\in V}\frac{vol(u,t_i,t_i^*)}{vol(u,t,t'')}\lambda_2(\mathcal{N}_i)\\
                            & = \uline{\phi_c}(G^\tau)
\end{align*}
\end{proof}\vsc

\eat{
\section*{Implementation details} 
 We conduct all experiments on a $3.6$GHz Intel processor with $16$GB of RAM. All algorithms are implemented in Java. To compute eigenvalues, we employ the implementation of the Lanczos algorithm from the \emph{Matrix Toolkit Java}\footnote{Matrix Toolkit Java. Obtained from \url{https://github.com/fommil/matrix-toolkits-java.}}.
\begin{figure}[t]
\centering
\includegraphics[width=0.5\textwidth]{fig/true_positives_multiple.pdf}
\caption{Effect of hashing parameters $b$ and $r$ on locating injected temporal community. }\vsb\vsb
\label{fig:br}
\end{figure}

\section*{Effect of parameters}
In order to test the dependence on hashing parameters, we vary the contrast of the injected community in our synthetic dataset and measure the ability of our hashing scheme to detect seeds in the injected community in Fig.~\ref{fig:br}. A bucket is a true positive (TP) if it contains only temporal nodes from the injected community. We plot the fraction of TP buckets, termed true positive rate (TPR), as a function of the number of independent minhash neighborhood signatures $k$ for different number of signatures $b$ composed using an OR predicate, referred to as number of bands~\cite{ullman_mining_book} ($k$ is fixed to the theoretical optimal $k^*=20$ for the length of the injected community). Each data point represents the average of $50$ runs. As expected, the higher number of bands increases the TPR, however, there is a diminishing returns effect past $7$ bands. In low-contrast settings the TPR does not exceed $50$ and peaks at a specific number of neighborhood minhashes $r$, past which the probability of exact signature match in all bands decreases.  The best number of minhashes $r$ increases with the contrast, and so does the TPR reaching levels above $0.6$. TPR initially increase with the number of rows $r$, but past a certain peak value the number of FP increases. Similar behavior of optimal number of rows has been demonstrated in minhash-based LSH in the unweighted case~\cite{ullman_mining_book}. 

\eat{
\noindent {\bf Effect of temporal pivots}

theoretical k* is useful in practice in targeting communities of given length.

\begin{figure}[t]
\centering
\subfigure[][Effect of temporal norm. $\alpha$]
{
 \centering
  \includegraphics[width=0.3\textwidth]{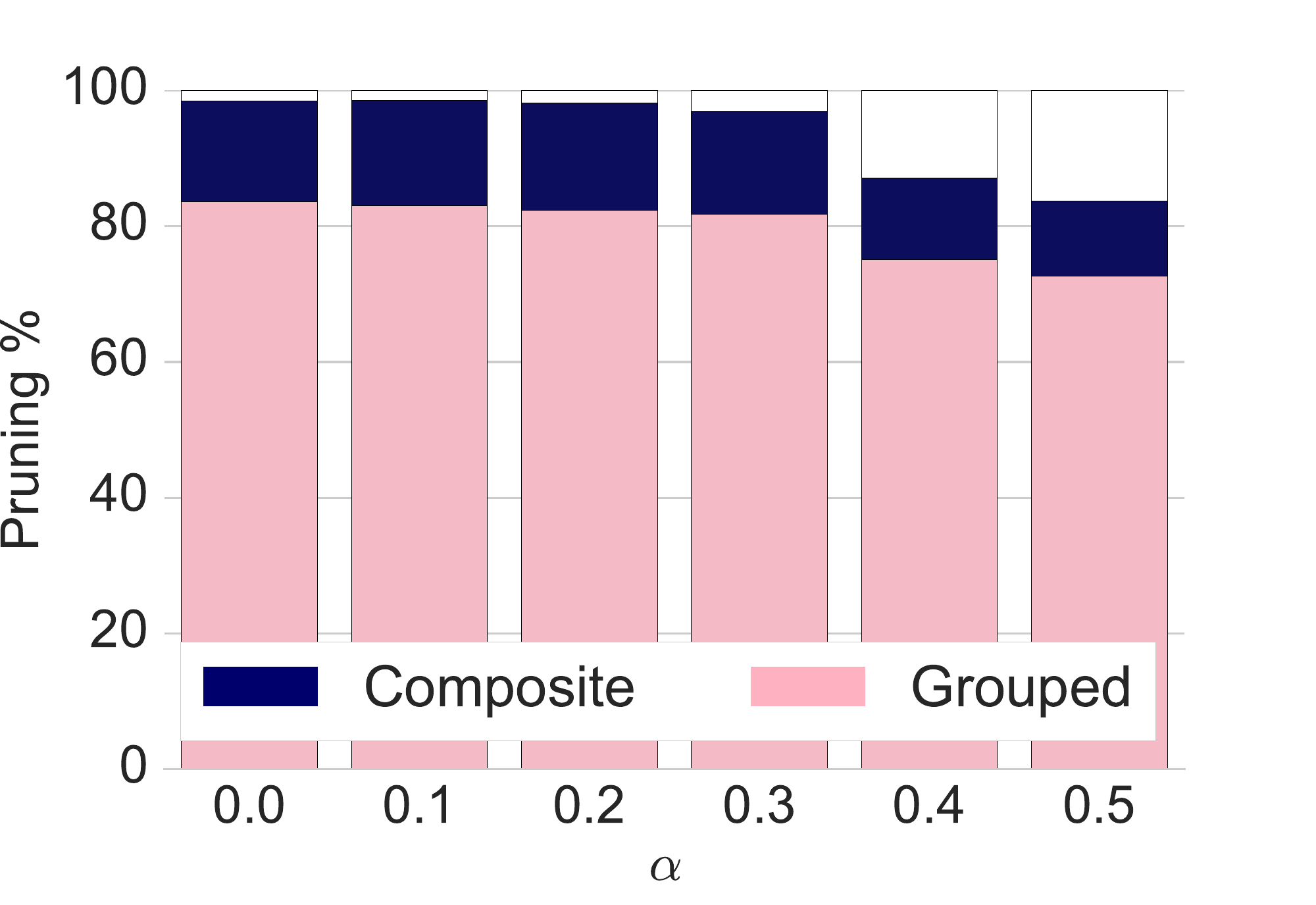}
  \label{fig:prunenorm}
}\hspace{-0.2in}
\subfigure[][Effect of temporal pivots $k$]
{
\centering
  \includegraphics[width=0.25\textwidth]{fig/pivots_truepos.pdf}
  \label{fig:tprk}
}
\caption{\subref{fig:prunenorm} Pruning power in synthetic data for different levels of temporal normalization $\alpha$. \subref{fig:tprk} TPR of collision buckets in the synthetic dataset as a function of the number of time pivots $k$ \todo{Redo figure}}
\end{figure}
}
}

\end{document}